\documentclass[11pt,onecolumn, draftcls]{IEEEtran}
\usepackage[T1]{fontenc}
\usepackage[latin9]{inputenc}
\usepackage{amsmath}
\usepackage{amssymb}
\usepackage{graphicx}

\makeatletter

\providecommand{\tabularnewline}{\\}

  \newtheorem{assumption}{Assumption}
  \newtheorem{definitn}{Definition}
  \newtheorem{remrk}{Remark}
  \newtheorem{problem}{Problem}
  \newtheorem{example}{Example}
  \newtheorem{lemma}{Lemma}
  \newtheorem{thm}{Theorem}
  \newtheorem{cor}{Corollary}

\usepackage{cite}
\usepackage[level = 0]{wgroup_message}

\makeatother

\begin{document}

\title{CSI Feedback Reduction for MIMO Interference Alignment }

\author{Xiongbin Rao, Liangzhong Ruan, \emph{Student Member, IEEE}, and \\Vincent
K.N. Lau, \emph{Fellow IEEE }%
\thanks{The authors are with the ECE Department of the Hong Kong University
of Science and Technology, Hong Kong (e-mails: \{xrao,stevenr,eeknlau\}@ust.hk).%
}}
\maketitle
\begin{abstract}
Interference alignment (IA) is a linear precoding strategy that can
achieve optimal capacity scaling at high SNR in interference networks.
Most of the existing IA designs require full channel state information
(CSI) at the transmitters, which induces a huge CSI signaling cost.
Hence it is desirable to improve the feedback efficiency for IA and
in this paper, we propose a novel IA scheme with a significantly reduced
CSI feedback. To quantify the CSI feedback cost, we introduce a novel
metric, namely the \emph{feedback dimension.} This metric serves as
a first-order measurement of CSI feedback overhead. Due to the partial
CSI feedback constraint, conventional IA schemes can not be applied
and hence, we develop a novel IA precoder / decorrelator design and
establish new IA feasibility conditions. Via dynamic \emph{feedback
profile} design, the proposed IA scheme can also achieve a flexible
tradeoff between the degree of freedom (DoF) requirements for data
streams, the antenna resources and the CSI feedback cost. We show
by analysis and simulations that the proposed scheme achieves substantial
reductions of CSI feedback overhead under the same DoF requirement
in MIMO interference networks. 
\end{abstract}

\section{Introduction}

Due to the broadcast nature of wireless communication, interference
is one of the most serious performance bottlenecks in modern wireless
networks. Conventional interference management schemes either treat
interference as noise or use channel orthogonalization to avoid interference.
However, these schemes are far from optimal  in general \cite{host2005multiplexing}.
Interference alignment (IA), which aligns the aggregate interference
from different transmitters (Txs) into a lower dimensional subspace
at each receiver (Rx), achieves the optimal capacity scaling with
respect to (w.r.t.) signal to noise ratio (SNR) under a broad range
of network topologies \cite{cadambe2008interference,jafar2007MIMOX}.
For instance, in a $K$-user MIMO interference channel with $N$ antennas
at each node, the IA processing achieves the throughput $\mathcal{O}(\frac{KN}{2}\cdot\log\textrm{SNR})$
\cite{jafar2010MNDoF}. This scaling law significantly dominates that
achieved by conventional orthogonalization schemes ($\mathcal{O}(N\cdot\log\textrm{SNR}$)).
As such, there has been a resurgence of research interest in IA. 

Classical IA designs \cite{gomadam2011distributed,peters2011cooperative,santamaria2010maximum}
assume all cross link channel state information are perfectly known
at the Tx side (CSIT). In practice, it is very difficult to obtain
very accurate CSIT estimation due to the limited feedback capacity
and the performance of IA is highly sensitive to CSIT error \cite{nosrat2010MIMO}.
This motivates the need to reduce the CSI feedback for IA in MIMO
interference networks. For instance, composite Grassmannian codebooks
are deployed in \cite{thukral2009interference,krishnamachari2009interference}
to quantize and feedback the CSI matrices for IA in MIMO interference
networks. Schemes which reduce feedback overhead by adapting to the
spatial/temporal correlation of CSI are proposed in \cite{rao2012limited,el2011grassmannian}.
While the aforementioned works try to quantize and feedback entire
CSI matrices, \cite{rezaee2012interference,rezaee2013csit} propose
more efficient schemes by exploiting an interesting fact that IA algorithms
do not need full knowledge of these matrices and hence CSI matrices
can be truncated before quantization. Besides these approaches, a
greedy algorithm is also proposed in \cite{de2012interference} to
reduce the size of the CSI submatrices feedback in MIMO interference
networks with single stream transmission.

In this paper, we propose a novel CSI feedback scheme (\emph{with
no quantization}) to reduce the CSI feedback cost in MIMO interference
networks under a given number of data streams (DoF) requirement. Instead
of CSI truncation in \cite{rezaee2012interference,rezaee2013csit,de2012interference},
we consider a more holistic set of CSI reduction strategies by selectively
feeding back the essential parts of the CSI knowledge to achieve the
IA interference nulling requirements for all the data streams. We
first define a novel metric, namely the \emph{feedback dimension,}
to quantify the cost of CSI feedback in interference networks. This
metric represents the sum dimension of the Grassmannian  manifolds
\cite{hirsch1976differential,dai2008quantization} that contain the
CSI feedback matrices. We will illustrate in Section II that this
metric serves as a first-order measurement of the CSI feedback overhead%
\footnote{The feedback dimension measures the amount of CSI information to feedback,
but it does not account for the quantization. As a result, it is proportional
to the total number of bits for CSI feedback in the interference networks.
Please refer to Section II for details.%
}. We consider IA design under the proposed partial CSI feedback scheme
and develop a novel precoder / decorrelator algorithm to achieve the
IA interference nulling in MIMO interference networks. By introducing
a dynamic \emph{feedback profile}%
\footnote{Feedback profile refers to a parametrization of the feedback functions
that determine how the CSI matrices are fed back to the Txs in the
interference networks. Please refer to Section II for details.%
} design, the proposed scheme achieves a flexible tradeoff between
the performance, i.e. DoFs, and the CSI feedback cost in interference
networks. To achieve these goals, there are several first order technical
challenges to tackle. 
\begin{itemize}
\item \textbf{Feedback Profile Design:} To reduce the CSI feedback cost,
only part of the CSI matrices can be fed back, but which part of the
CSI matrices to feedback (feedback profile design) is a challenging
problem. As illustrated in Example 1, a good feedback profile design
can significantly reduce the feedback cost to achieve IA in interference
networks. The feedback profile design in interference networks is
not widely studied in the literature. In \cite{cho2011feedback},
the authors propose a two-hop centralized feedback profile, but the
framework relies heavily on closed form precoder solutions for IA.
As such, the approach in \cite{cho2011feedback} can only be applied
to very limited interference network topologies. In general, the feedback
profile design is combinatorial and is very challenging. 
\item \textbf{IA Feasibility Condition:} Given a number of antennas and
data streams, the IA problem is well known to be not always feasible,
and the feasibility condition is still not fully understood in general.
The pioneering work \cite{yetis2010feasibility} gives the feasibility
condition for the single stream case by using the \textit{Bernshtein\textquoteright{}s
Theorem }\cite{cox2005using}. This work is extended to the multiple
stream case in \cite{razaviyayn2011degrees} by analyzing the dimension
of the \emph{Algebraic Varieties} \cite{cox2005using}. In \cite{ruan2012feasibility},
a sufficient feasibility condition, which applies to general MIMO
interference networks, is proposed. However, all these existing works
have assumed feedback of entire CSI matrices in the interference networks.
The feasibility condition of IA under partial\emph{ }CSI feedback
in interference networks is still an open problem. 
\item \textbf{IA Precoder / Decorrelator Design:} Conventional IA precoder
algorithms \cite{gomadam2011distributed,peters2011cooperative,santamaria2010maximum}
require full CSI matrices of the interference networks, and both the
precoders and decorrelators are functions of the entire CSI matrices
in the MIMO interference networks. However, to reduce the CSI feedback
cost, only partial CSI matrices will be available at the Txs and hence,
the precoders can only be a function of the \emph{partial CSI}. As
a result, conventional solutions for IA precoder and decorrelator
designs cannot be applied in our case.
\end{itemize}

In this paper, we will address the challenges listed above by exploiting
the unique features of the IA precoder / decorrelator design, and
tools from \emph{Algebraic Geometry} \cite{razaviyayn2011degrees,ruan2012feasibility},
to reduce the CSI feedback cost without affecting the DoF performance
of the network. Based on the proposed interference profile design
mechanism, we derive closed form tradeoff results between the number
of data streams, the antenna configuration and the CSI feedback dimension
in a symmetric MIMO interference network. We also show that the proposed
scheme achieves significant savings in CSI feedback cost compared
with various state-of-the-art baselines. 

\textit{Notation}s: Uppercase and lowercase boldface denote matrices
and vectors respectively. The operators $(\cdot)^{T}$, $(\cdot)^{\dagger}$,
$\textrm{rank}(\cdot)$, $|\cdot|$, $\textrm{tr}(\cdot)$, $\textrm{dim}_{s}(\cdot)$,
$\textrm{dim}_{c}(\cdot)$, $\otimes$ and $\textrm{vec}(\cdot)$
are the transpose, conjugate transpose, rank, cardinality, trace,
dimension of subspace, dimension of complex manifolds \cite{hirsch1976differential},
Kronecker product and vectorization respectively, $\mathbf{I}_{d}$
denotes $d\times d$ identity matrix, $\textrm{span}(\{\mathbf{A}_{i}\})$
denotes the vector space spanned by all the column vectors of the
matrices in $\{\mathbf{A}_{i}\}$ and $d\mid M$ denotes that integer
$M$ is divisible by integer $d$.

\section{System Model}

\subsection{MIMO Interference Networks}

Consider a $K$-user MIMO interference network where the $i$-th Tx
and Rx are equipped with $N_{i}$ and $M_{i}$ antennas respectively,
and $d_{i}$ data streams are transmitted between the $i$-th Tx-Rx
pair. Denote the fading matrix from Tx $i$ to Rx $j$ as $\mathbf{H}_{ji}\in\mathbb{C}^{M_{j}\times N_{i}}$
.
\begin{assumption}
[Channel Matrices]\label{Channel-MatricesWe-assume}We assume the
elements of $\mathbf{H}_{ji}$ are i.i.d. Gaussian random variables
with zero mean and unit variance. The CSI $\{\mathbf{H}_{j1},\mathbf{H}_{j2},\cdots\mathbf{H}_{jK}\}$
are observable at the $j$-th Rx and the feedback from the $j$-th
Rx will be received error-free by all the $K$ Txs. 
\end{assumption}

\subsection{CSI Feedback Functions and Feedback Dimension}

In this section, we define the partial CSI feedback as well as the
notion of \emph{feedback dimension} in MIMO interference networks.
Since we are interested in IA, which aims at nulling off interferences
between the data streams in the network, only the \emph{channel direction
information}%
\footnote{For example, in IA designs, if $\mathbf{\mathbf{U}}^{\dagger}\mathbf{H}\mathbf{V}=\mathbf{0}$,
then we have $\mathbf{\mathbf{U}}^{\dagger}(a\mathbf{H})\mathbf{V}=0$,
$\forall a\in\mathbb{C}$. Hence, it is sufficient to feeding back
the\emph{ direction information} of $\mathbf{H}\in\mathbb{C}^{N\times M}$
for IA, i.e., $\{a\mathbf{H}:a\in\mathbb{C}\}$, which is a linear
space contained in $\mathbb{G}(1,MN)$ \cite{dai2008quantization}.%
} \cite{yoo2007multi} is relevant, and hence, we restrict ourselves
to the CSI feedback over the Grassmannian manifold. Denote $\mathbb{G}(A,B)$
as the Grassmannian manifold \cite{dai2008quantization} of all $A$-dimensional
linear subspaces in $\mathbb{C}^{B\times1}$. Let $\mathcal{H}_{j}=(\mathbf{H}_{j1},\cdots\mathbf{H}_{jj-1},\mathbf{H}_{jj+1},\cdots\mathbf{H}_{jK})\in\prod_{\underset{\neq j}{i=1}}^{K}\mathbb{C}^{M_{j}\times N_{i}}$
be the tuple of \emph{local cross-link CSI matrices} observed at the
$j$-th Rx in the MIMO interference network. To reduce CSI feedback
overhead, we introduce the idea of CSI filtering, which is formulated
in the following model. 
\begin{definitn}
[CSI Feedback Function]The partial CSI feedback generated by the
$j$-th Rx is a $k_{j}$-tuple, which  can be characterized by a \emph{feedback
function} $F_{j}$: $\prod_{\underset{\neq j}{i=1}}^{K}\mathbb{C}^{M_{j}\times N_{i}}\rightarrow\prod_{i=1}^{k_{j}}\mathbb{G}(A_{ji},B_{ji})$.
That is:
\begin{equation}
\mathbb{\mathcal{H}}_{j}^{fed}=F_{j}(\mathcal{H}_{j}),\label{eq:feedback-info}
\end{equation}
where $k_{j}$ denotes the number of subspaces in $\mathbb{\mathcal{H}}_{j}^{fed}$,
$\mathbb{\mathcal{H}}_{j}^{fed}\in\mathbb{G}(A_{j1},B_{j1})\times\mathbb{G}(A_{j2},B_{j2})\times\cdots\mathbb{G}(A_{jk_{j}},B_{jk_{j}})$
is the partial CSI fed back by the $j$-th Rx, and $\mathbb{G}(A_{jm},B_{jm})$
is the associated Grassmannian manifold with parameters $(A_{jm},B_{jm})$
containing the $m$-th element in the CSI feedback tuple $\mathbb{\mathcal{H}}_{j}^{fed}$.
\hfill \IEEEQED
\end{definitn}

In other words, the outputs of the feedback function are a tuple of
subspaces where each subspace corresponds to a point in the associated
Grassmannian manifold \cite{dai2008quantization}. For instance, consider
two cross link CSIs $\mathbf{H}_{1},\mathbf{H}_{2}\in\mathbb{C}^{2\times3}$
at certain Rx. If we feedback the null spaces of $\mathbf{H}_{1}$
and $\mathbf{H}_{2}$, then this corresponds to the feedback function
$F=\left(\begin{array}{cc}
\{\mathbf{v}_{1}:\mathbf{H}_{1}\mathbf{v}_{1}=\mathbf{0}\}, & \{\mathbf{v}_{2}:\mathbf{H}_{2}\mathbf{v}_{2}=\mathbf{0}\}\end{array}\right)\in\mathbb{G}(1,3)\times\mathbb{G}(1,3)$; If we feedback the row space of the concatenated matrix $\begin{array}{cc}
[\mathbf{H}_{1} & \mathbf{H}_{2}]\end{array}$, this corresponds to the feedback function $F=\textrm{span}\left([\begin{array}{cc}
\mathbf{H}_{1} & \mathbf{H}_{2}\end{array}]^{T}\right)\in\mathbb{G}(2,6)$. Note under \emph{given} feedback functions $\{F_{j}\}$, the partial
CSI $\{F_{j}(\mathcal{H}_{j})\}$ that is fed back to the Tx side
for precoder design will be known.

First, we define the feedback cost generated from the above partial
CSI feedback by the \emph{feedback dimension} below.
\begin{definitn}
[Feedback Dimension]\label{Feedback-DimensionDefine-the}Define
the feedback dimension $D$ as the sum of the dimension of the Grassmannian
manifolds \cite{dai2008quantization} $\{\mathbb{G}(A_{ji},B_{ji}):i=1,\cdots k_{j},j=1,\cdots K\}$,
i.e., 
\begin{equation}
D=\sum_{j=1}^{K}\sum_{i=1}^{k_{j}}A_{ji}(B_{ji}-A_{ji}).\label{eq:definition_feedback_dimensioin}
\end{equation}
 \hfill \IEEEQED
\end{definitn}

\begin{remrk}
[Significance of Feedback Dimension]Note a Grassmannian manifold
of dimension $D$ is locally homeomorphic  \cite{hirsch1976differential}
to $\mathbb{C}^{D\times1}$, and hence the feedback dimension $D$
denotes the number of complex scalars required to feedback to the
Tx side. Hence, the feedback dimension serves as a first order metric
of the CSI feedback overhead. For instance, given $B$ bits to feedback
a CSI contained in a Grassmannian manifold with dimension $D$, it
is shown that the CSI quantization distortion scales on $\mathcal{O}\left(2^{-\frac{B}{D}}\right)$
\cite{mondal2007quantization,dai2008quantization}. In other words,
to keep a constant CSI distortion $\Delta$, the CSI feedback bits
$B$ should scale linearly with $D$ as $B=\mathcal{O}(D\log\frac{1}{\Delta})$.
Therefore, the feedback dimension is directly proportional to the
total number of bits for CSI feedback. 
\end{remrk}

\subsection{CSI Feedback Profile}

In this section, we shall define the notion of \emph{feedback profile},
which is a parametrization of the feedback functions $\mathcal{F}=\{F_{1},\cdots F_{K}\}$
defined in (\ref{eq:feedback-info}). We first formally define the
IA problem subject to general feedback functions $\mathcal{F}$, which
is essentially a \emph{feasibility}%
\footnote{We are concerned with the existence of a solution in Problem 1 as
well as finding it.%
}\emph{ problem} \cite{razaviyayn2011degrees,ruan2012feasibility}. 
\begin{problem}
[IA Design with Partial CSI Feedback $\mathcal{F}$]\label{IA-design-based-general}Given
the feedback functions $\mathcal{F}$. The IA problem is to find the
set of precoders $\left\{ \mathbf{V}_{i}\in\mathbb{C}^{N_{i}\times d_{i}}:\forall i\right\} $
as a function of $\{F_{j}(\mathcal{H}_{j}):\forall j\}$ and decorrelator
$\mathbf{U}_{j}\in\mathbb{C}^{M_{j}\times d_{j}}$ based on local
CSIR (i.e., $\{\mathbf{H}_{ji}\mathbf{V}_{i}:\forall i\}$) $\forall j$
such that 
\begin{equation}
\textrm{rank}(\mathbf{U}_{j}^{\dagger}\mathbf{H}_{jj}\mathbf{V}_{j})=d_{j},\quad\forall j,\label{eq:MIA}
\end{equation}
\begin{equation}
\mathbf{U}_{j}^{\dagger}\mathbf{H}_{ji}\mathbf{V}_{i}=\mathbf{0},\;\forall i,j,\; i\neq j.\label{eq:IA_condition}
\end{equation}
 \hfill \IEEEQED
\end{problem}

Compared with conventional IA problems \cite{yetis2010feasibility,razaviyayn2011degrees,ruan2012feasibility},
Problem \ref{IA-design-based-general} is different and difficult
because it has a new constraint on the available CSI knowledge for
precoder design, i.e., $\{\mathbf{V}_{i}\}$ can only be functions
of the partial CSI $\{F_{j}\}$ that is fed back. This reflects the
motivation to reduce the CSI feedback cost while maintaining the IA
performance in MIMO interference networks. In most conventional works
of feedback designs for IA in MIMO interference network \cite{rao2012limited,el2011grassmannian},
it has been considered that the \emph{full channel direction} is fed
back, (i.e., $F_{j}(\mathcal{H}_{j})=\left(\begin{array}{ccc}
\cdots, & \{a\mathbf{H}_{ji}:a\in\mathbb{C}\}, & \cdots\end{array}\right)_{i\neq j}$ ), which corresponds to a feedback dimension of $\sum_{i,j:i\neq j}(M_{j}N_{i}-1)$
in the MIMO interference networks. In the case of full channel direction
feedback, the solution to Problem \ref{IA-design-based-general} has
been widely studied \cite{gomadam2011distributed,peters2011cooperative,santamaria2010maximum},
and under some sufficient conditions \cite{yetis2010feasibility,razaviyayn2011degrees,ruan2012feasibility},
Problem \ref{IA-design-based-general} above is feasible. However,
the challenge comes when the CSI direction are not fully fed back. 

Yet, for a \emph{given} number of DoF requirements and antennas setups
in MIMO interference networks, the full CSI direction might not always
be required while IA can still be achieved. As illustrated by three
motivating examples in Figure 1 (a)-(c), we show that Problem \ref{IA-design-based-general}
can still be feasible with a much smaller feedback dimension. Denote
$\mathbb{N}^{t}(\cdot)$ and $\mathbb{N}^{r}(\cdot)$ as the null
space and left null space respectively, i.e., $\mathbb{N}^{t}(\mathbf{A})=\{\mathbf{u}\mid\mathbf{A}\mathbf{u}=\mathbf{0}\}$,
$\mathbb{N}^{r}(\mathbf{A})=\{\mathbf{u}\mid\mathbf{u}^{\dagger}\mathbf{A}=\mathbf{0}\}$. 
\begin{example}
[CSI Feedback Design I]Consider a MIMO interference network as illustrated
in Fig. \ref{fig:Example-of-feedback} (a). Suppose the CSI feedback
functions are given by:  $F_{1}(\mathcal{H}_{1})=\mathbb{N}_{1}^{t}(\mathbf{H}_{12})$,
$F_{2}(\mathcal{H}_{2})=\mathbb{N}^{t}(\mathbf{H}_{23})$, $F_{3}(\mathcal{H}_{3})=\mathbb{N}^{t}(\mathbf{H}_{31})$.
 The precoders $\mathbf{V}_{1}$, $\mathbf{V}_{2}$ , $\mathbf{V}_{3}$$\in\mathbb{C}^{6\times2}$
are designed as: $\textrm{span}(\mathbf{V}_{1})=\mathbb{N}_{1}^{t}(\mathbf{H}_{31})$,
$\textrm{span}(\mathbf{V}_{2})=\mathbb{N}^{t}(\mathbf{H}_{12})$,
$\textrm{span}(\mathbf{V}_{3})=\mathbb{N}^{t}(\mathbf{H}_{23})$,
the decorrelators $\mathbf{U}_{1}$, $\mathbf{U}_{2}$ , $\mathbf{U}_{3}$$\in\mathbb{C}^{4\times2}$
are designed as: $\textrm{span}(\mathbf{U}_{1})=\mathbb{N}^{r}\left(\mathbf{H}_{13}\mathbf{V}_{3}\right)$,
$\textrm{span}(\mathbf{U}_{2})=\mathbb{N}^{r}\left(\mathbf{H}_{21}\mathbf{V}_{1}\right)$,
$\textrm{span}(\mathbf{U}_{3})=\mathbb{N}^{r}\left(\mathbf{H}_{32}\mathbf{V}_{2}\right)$.
Consequently, Problem 1 is almost surely feasible, and the feedback
dimension is only 24 compared with 138 under full channel direction
feedback. 
\end{example}

\begin{example}
[CSI Feedback Design II]Consider a MIMO interference network as
illustrated in Fig. \ref{fig:Example-of-feedback} (b). Suppose the
CSI feedback functions are given by: $F_{1}(\mathcal{H}_{1})=\mathbb{N}^{t}\left((\mathbf{S}_{1}^{r})^{\dagger}\mathbf{H}_{12}\right)$,
$F_{2}(\mathcal{H}_{2})=\mathbb{N}^{t}\left((\mathbf{S}_{2}^{r})^{\dagger}\mathbf{H}_{21}\right)$,
$F_{3}(\mathcal{H}_{3})=\left(\mathbb{N}^{t}(\mathbf{H}_{32}),\mathbb{N}^{t}(\mathbf{H}_{31})\right)$,
where $\textrm{span}(\mathbf{S}_{1}^{r})=\mathbb{N}^{r}(\mathbf{H}_{13})$
and $\textrm{span}(\mathbf{S}_{2}^{r})=\mathbb{N}^{r}(\mathbf{H}_{23})$.
The precoders are designed as: $\textrm{span}(\mathbf{V}_{1})=\mathbb{N}^{t}\left((\mathbf{S}_{2}^{r})^{\dagger}\mathbf{H}_{21}\right)\bigcap\mathbb{N}^{t}(\mathbf{H}_{31})$,
$\textrm{span}(\mathbf{V}_{2})=\mathbb{N}^{t}\left((\mathbf{S}_{1}^{r})^{\dagger}\mathbf{H}_{12}\right)\bigcap\mathbb{N}^{t}(\mathbf{H}_{32})$
and $\mathbf{V}_{3}=\mathbf{I}_{2}$. Problem 1 is also almost surely
feasible, and the feedback dimension is only 32 compared with 82 under
under full channel direction feedback. 
\end{example}

\begin{example}
[CSI Feedback Design III]Consider a MIMO interference network as
illustrated in Fig. \ref{fig:Example-of-feedback} (c). Suppose the
CSI feedback functions are given by: $F_{1}(\mathcal{H}_{1})=\textrm{span}\left(\left[\begin{array}{cc}
\mathbf{H}_{12}^{s} & \mathbf{H}_{13}^{s}\end{array}\right]^{T}\right)$, $F_{2}(\mathcal{H}_{2})=\textrm{span}\left(\left[\begin{array}{cc}
\mathbf{H}_{21}^{s} & \mathbf{H}_{23}^{s}\end{array}\right]^{T}\right)$ and $F_{3}(\mathcal{H}_{3})=\textrm{span}\left(\left[\begin{array}{cc}
\mathbf{H}_{31}^{s} & \mathbf{H}_{32}^{s}\end{array}\right]^{T}\right)$, where $\mathbf{H}_{ji}^{s}=\left[\begin{array}{cc}
\mathbf{I}_{4} & \mathbf{0}\end{array}\right]\mathbf{H}_{ji}$, $\forall j,i.$ Problem 1 is also almost surely feasible, and the
feedback dimension is only 48 compared with 114 under full channel
direction feedback.
\end{example}
\begin{figure*}
\begin{centering}
\includegraphics[scale=0.7]{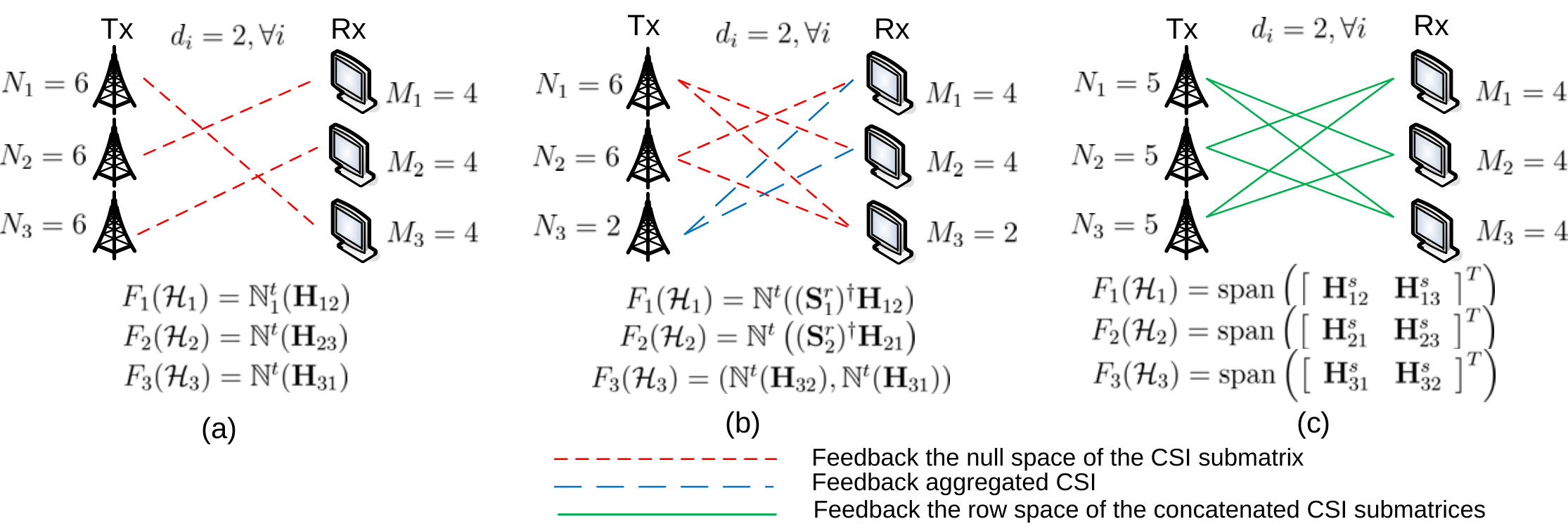}
\par\end{centering}

\vspace{0.5cm}

\begin{centering}
\begin{tabular}{|c|c|c|c|}
\hline 
Feedback Dimension & In Fig. \ref{fig:Example-of-feedback}(a) & In Fig. \ref{fig:Example-of-feedback}(b) & In Fig. \ref{fig:Example-of-feedback}(c)\tabularnewline
\hline 
\hline 
Full Channel Direction Feedback & 138 & 82 & 114\tabularnewline
\hline 
Proposed Feedback Schemes & 24 & 32 & 48\tabularnewline
\hline 
\end{tabular}
\par\end{centering}

\caption{\label{fig:Example-of-feedback}Example of feedback topology design}
\end{figure*}

In the above three examples, Problem \ref{IA-design-based-general}
is feasible even if the total feedback dimension are 24, 32 and 48
respectively. This represents an 83, 69 and 58 \% reduction in the
feedback cost compared with full channel direction feedback. The following
four insights can be obtained from these three examples on how to
reduce the feedback dimension at each Rx:
\begin{itemize}
\item \textbf{Strategy I (No Feedback for a Subset of Cross Links)}: In
practice, IA may be achieved with no feedback for a subset of the
cross links. For instance, in Example 1, cross links $\mathbf{H}_{13}$,
$\mathbf{H}_{21}$ and $\mathbf{H}_{32}$ are not fed back at all.
With this strategy, Problem 1 is still feasible and the feedback dimension
is significantly reduced.
\item \textbf{Strategy II (Feedback of Aggregate CSI for a Subset of Cross
Links)}: In practice, IA may be achieved by feeding back the aggregate
CSI for a subset of cross links. For instance, in Example 2, link
$\mathbf{H}_{13}$ is canceled by designing the decorrelator of Rx
1 in the space of $\mathbb{N}^{r}(\mathbf{H}_{13})$. Hence, the necessary
feedback information for that link is $\mathbf{S}_{1}^{r}$ ($\textrm{span}(\mathbf{S}_{1}^{r})=\mathbb{N}^{r}(\mathbf{H}_{13})$),
which is aggregated in the feedback CSI of the other subsets of cross
links (e.g., the CSI feedback for the link from Tx 2 to Rx 1 has the
form $(\mathbf{S}_{1}^{r})^{\dagger}\mathbf{H}_{12}$). With this
strategy, Problem 1 is still feasible and the feedback dimension is
reduced.
\item \textbf{Strategy III (Feedback of Null Space of CSI Submatrix for
a Subset of Cross Links)}: In practice, IA may be achieved by feeding
back the null spaces for a subset of cross links. This is because
the Tx can design the precoder in the channel null space to cancel
that link. For instance, in Example 2, only the null spaces of $(\mathbf{S}_{1}^{r})^{\dagger}\mathbf{H}_{12}$
are fed back at Rx 1. With this strategy, Problem 1 is still feasible
and the feedback dimension is reduced.
\item \textbf{Strategy IV (Feedback of Row Space of CSI Submatrices} \textbf{for
a Subset of Cross Links}): In practice, IA can be achieved by feeding
back the row space of the concatenated CSI submatrices for a subset
of cross links. For instance, in Example 3, only $\textrm{span}\left(\left[\begin{array}{cc}
\mathbf{H}_{12}^{s} & \mathbf{H}_{13}^{s}\end{array}\right]^{T}\right)$ is fed back at Rx 1. With this strategy, Problem 1 is still feasible
and the feedback dimension is reduced.
\end{itemize}

Note that it is possible to use only one of the above strategies or
apply them together and how to use these strategies depends on the
DoF requirements and the antenna configurations. Furthermore, different
combinations of these strategies may have significantly different
IA feasibility result and final feedback cost. To begin with, we assume
some structure forms for the feedback functions that can embrace all
these 4 strategies. Based on the above insights, we shall first partition
the cross links seen by the $j$-th Rx into four subsets defined below.
\begin{definitn}
[Partitioning of Cross Links]\label{Partitioning-of-Cross}The set
of cross links seen by the $j$-th Rx is partitioned into four subsets,
namely, $\Omega_{j}^{I}$, $\Omega_{j}^{II}$, $\Omega_{j}^{III}$
and $\Omega_{j}^{IV}$, according to the four strategies illustrated
above. Note that $\underset{m=\{I,II,III,IV\}}{\mathbf{\bigcup}}\Omega_{j}^{m}=\{1,\cdots j-1,j+1,\cdots K\}$,
$\Omega_{j}^{m}\bigcap\Omega_{j}^{n}=\emptyset,\forall m\neq n,m,n\in\{I,II,III,IV\}$.\hfill \IEEEQED
\end{definitn}

The feedback functions $\mathcal{F}$ are assumed to have the following
structure.

\begin{figure}
\begin{centering}
\includegraphics[scale=0.8]{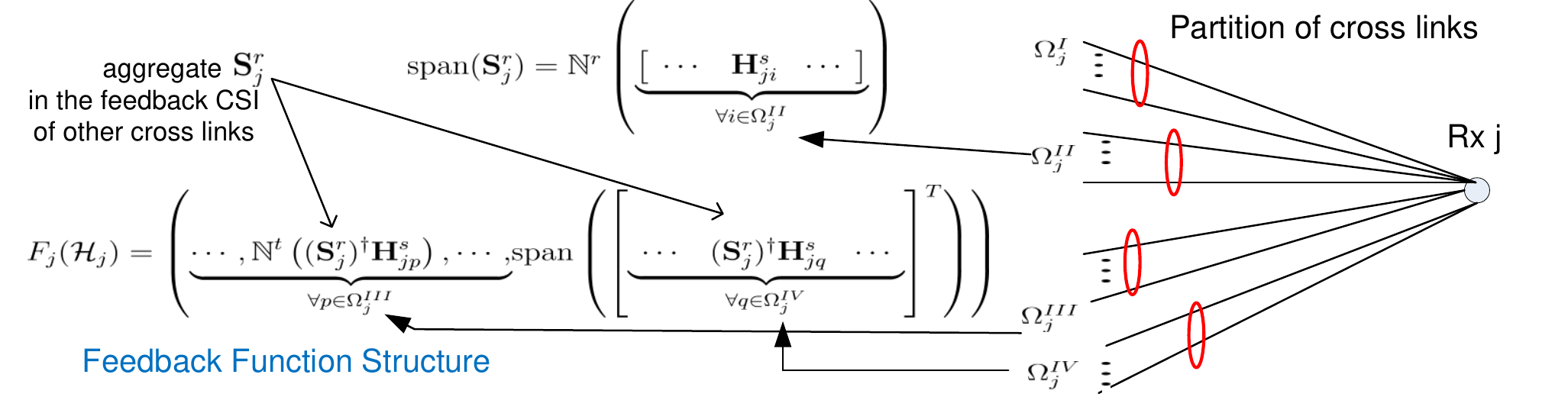}
\par\end{centering}

\caption{\label{fig:Illustration-of-Feedback}Illustration of Feedback Function
structure}
\end{figure}

\begin{assumption}
[Structure of Feedback Functions $\mathcal{F}$]\label{Structure-of-Feedback}The
feedback functions $\mathcal{F}_{j}$ in (\ref{eq:feedback-info})
for the MIMO interference networks have the following structure:
\begin{equation}
F_{j}(\mathcal{H}_{j})=\left(\underset{\forall p\in\Omega_{j}^{III}}{\underbrace{\cdots,\mathbb{N}^{t}\left((\mathbf{S}_{j}^{r})^{\dagger}\mathbf{H}_{jp}^{s}\right),\cdots,}}\textrm{span}\left(\left[\underset{\forall q\in\Omega_{j}^{IV}}{\underbrace{\begin{array}{ccc}
\cdots & (\mathbf{S}_{j}^{r})^{\dagger}\mathbf{H}_{jq}^{s} & \cdots\end{array}}}\right]^{T}\right)\right),\label{eq:feedback_structure}
\end{equation}
where $\mathbf{S}_{j}^{r}\in\mathbb{C}^{M_{j}^{s}\times M_{j}^{e}}$,
$(\mathbf{S}_{j}^{r})^{\dagger}\mathbf{S}_{j}^{r}=\mathbf{I}_{M_{j}^{e}}$,
\begin{equation}
\textrm{span}(\mathbf{S}_{j}^{r})=\mathbb{N}^{r}\left(\underset{\forall i\in\Omega_{j}^{II}}{\underbrace{\left[\begin{array}{ccc}
\cdots & \mathbf{H}_{ji}^{s} & \cdots\end{array}\right]}}\right),\label{eq:s_j^r_infor}
\end{equation}
\begin{equation}
M_{j}^{e}=M_{j}^{s}-\sum_{i\in\Omega_{j}^{II}}N_{i}^{s},\forall j.\label{eq:notaion_M_j^e}
\end{equation}
\begin{equation}
\mathbf{H}_{ji}^{s}=\left[\begin{array}{cc}
\mathbf{I}_{M_{j}^{s}} & \mathbf{0}\end{array}\right]\mathbf{H}_{ji}\left[\begin{array}{c}
\mathbf{I}_{N_{j}^{s}}\\
\mathbf{0}
\end{array}\right],\forall j,i.\label{eq:sub_channel_matrix}
\end{equation}
and $\{M_{i}^{s},N_{i}^{s}\}$ are parameters that characterize the
feedback functions $\mathcal{F}$. \hfill \IEEEQED
\end{assumption}

The feedback function structure is also illustrated in Fig. \ref{fig:Illustration-of-Feedback}.
Note that the length of the tuple $\mathbb{\mathcal{H}}_{j}^{fed}$
is $k_{j}=|\Omega_{j}^{III}|+1$. Denote $\Omega_{j}^{III}\triangleq\{p_{1},\cdots,p_{i},\cdots\}$,
then $\mathbb{N}^{t}\left((\mathbf{S}_{j}^{r})^{\dagger}\mathbf{H}_{jp_{i}}^{s}\right)\in\mathbb{G}(A_{ji},B_{ji})$,
where $A_{ji}=M_{j}^{e}$, $B_{ji}=N_{p_{i}}^{s}$, $1\leq i\leq|\Omega_{j}^{III}|$
and $\textrm{span}\left(\left[\begin{array}{ccc}
\cdots & (\mathbf{S}_{j}^{r})^{\dagger}\mathbf{H}_{jq}^{s} & \cdots\end{array}\right]_{\forall q\in\Omega_{j}^{IV}}^{T}\right)\in\mathbb{G}(A_{jk_{j}},B_{jk_{j}})$, where $A_{jk_{j}}=\min\left(M_{j}^{e},\sum_{v\in\Omega_{j}^{IV}}N_{v}^{s}\right)$,
$B_{jk_{j}}=\sum_{v\in\Omega_{j}^{IV}}N_{v}^{s}$, as in (\ref{eq:feedback-info}). 

Note that the structural form of $\mathcal{F}$ in (\ref{eq:feedback_structure})
embraces all four strategies  of CSI feedback dimension reduction
inspired by examples 1-3. Based on the structural form of $\mathcal{F}$,
we define the notion of \emph{feedback profile} of $\mathcal{F}$,
which gives a parametrization of $\mathcal{F}$. 
\begin{definitn}
[Feedback Profile of $\mathcal{F}$]\label{Feedback-Profile-of}Define
the feedback profile of $\mathcal{F}$ as a set of parameters:
\begin{equation}
\mathcal{L}=\left\{ \left\{ M_{i}^{s},N_{i}^{s}:\forall i\right\} ,\left\{ \Omega_{j}^{I},\Omega_{j}^{II},\Omega_{j}^{III},\Omega_{j}^{IV}:\forall j\right\} \right\} .\label{eq:feedback_topology}
\end{equation}
\hfill \IEEEQED
\end{definitn}

Note that $\{M_{j}^{s},N_{i}^{s}\}$ controls the size of the CSI
submatrices to feedback and $\left\{ \Omega_{j}^{m}:m\in\{I,II,III,IV\}\}\right\} $
defines the partitioning of the cross links w.r.t. the four feedback
strategies at the $j$-th Rx. In fact, there is a 1-1 correspondence
between the feedback profile $\mathcal{L}$ and the feedback function
in (\ref{eq:feedback_structure}). For a given feedback profile $\mathcal{L}$
(or feedback function $\mathcal{F}$), the total feedback dimension
is given by,
\begin{equation}
D(\mathcal{L})=\sum_{j=1}^{K}M_{j}^{e}\left(\sum_{i\in\Omega_{j}^{IV}}N{}_{i}^{s}-M_{j}^{e}\right)^{+}+\sum_{j=1}^{K}\sum_{i\in\Omega_{j}^{III}}M_{j}^{e}\left(N{}_{i}^{s}-M_{j}^{e}\right).\label{eq:sum_feedback_dimension_expression}
\end{equation}

In fact, the CSI feedback function in (\ref{eq:s_j^r_infor}) and
the associated feedback profile in (\ref{eq:feedback_topology}) cover
a lot of existing CSI feedback designs in the literature, and we mention
a few below. 
\begin{itemize}
\item \textbf{Special case I }\emph{(Feedback Truncated CSI)}: In \cite{rezaee2012interference,rezaee2013csit},
a truncated CSI feedback scheme is proposed in MIMO interference network.
The feedback scheme corresponds to the feedback profile $\mathcal{L}$
: $M_{i}^{s}=M_{i}$, $N_{i}^{s}=N_{i}$, $\Omega_{j}^{I}=\Omega_{j}^{II}=\Omega_{j}^{III}=\emptyset$,
$\Omega_{j}^{IV}=\{1,\cdots,j-1,j+1,\cdots K\}$, $\forall j$, and
feedback function $F_{j}=\textrm{span}\left(\left[\begin{array}{ccc}
\cdots & \mathbf{H}_{ji} & \cdots\end{array}\right]_{i:\, i\neq j}^{T}\right)$, $\forall j$.
\item \textbf{Special case II} \emph{(Two-hop Centralized CSI Feedback)}:
In \cite{cho2011feedback}, a centralized two-hop feedback scheme
is proposed based on the closed form solutions of IA in MIMO interference
network. The feedback scheme corresponds to the feedback profile $\mathcal{L}$:
$M_{i}^{s}=M_{i}$, $N_{i}^{s}=N_{i}$, $\Omega_{j}^{II}=\Omega_{j}^{III}=\emptyset$,
$\Omega_{j}^{IV}=\{a_{j},b_{j}\}$, $\Omega_{j}^{I}=\{1,\cdots,j-1,j+1,\cdots K\}/\Omega_{j}^{IV}$,
$\forall j$, where $(a_{1},b_{1})=(2,3),$ $(a_{2},b_{2})=(3,4),$...$(a_{K-1},b_{K-1})=(K,1)$,
$(a_{K},b_{K})=(1,2)$, and the feedback function $F_{j}=\textrm{span}\left(\left[\begin{array}{cc}
\mathbf{H}_{ja_{j}} & \mathbf{H}_{jb_{j}}\end{array}\right]^{T}\right)=\textrm{span}\left(\left[\mathbf{H}_{jb_{j}}^{-1}\begin{array}{cc}
\mathbf{H}_{ja_{j}} & \mathbf{I}\end{array}\right]^{T}\right)$, $\forall j$.
\end{itemize}

\section{Design of IA Precoders and Decoders under a Feedback Profile}

In this section, we focus on solving the IA precoders and decorrelators
design in Problem 1 for a given feedback profile $\mathcal{L}$. Specifically,
we first impose some structural properties on the precoders / decorrelators
so as to satisfy the constraints of partial CSI feedback. Based on
the proposed structures, we transform Problem 1 into an equivalent
bi-convex problem and derive an iterative solution.

\subsection{Structure of IA Precoders / Decorrelators}

One unique challenge of the IA precoders / decorrelators design in
Problem 1 is that the precoders can only be adaptive to the partial
CSI knowledge at the Txs. This is fundamentally different from conventional
IA precoders / decorrelators design in which both can be adaptive
to the entire CSI matrices. To address this challenge, we shall first
impose some structures in the precoders / decorrelators so as to utilize
the partial CSI obtained from combinations of feedback strategies%
\footnote{Notice that Strategy I does not feedback the CSI for the chosen subset
of cross links.%
} I-IV. 
\begin{itemize}
\item \textbf{Utilization of Partial CSI from Feedback Strategy II}: From
feedback strategy II, we can obtain the aggregated CSI with $\mathbf{S}_{j}^{r}$,
which spans $\mathbb{N}^{r}\left(\left[\begin{array}{ccc}
\cdots & \mathbf{H}_{ji}^{s} & \cdots\end{array}\right]_{\forall i\in\Omega_{j}^{II}}\right)$. Hence, we can design the decorrelator of Rx $j$ in the space of
$\textrm{span}(\mathbf{S}_{j}^{r})$, and consequently, all interference
from Tx $i$ to Rx $j$, where $i\in\Omega_{j}^{II}$, is eliminated
(note $(\mathbf{S}_{j}^{r})^{\dagger}\mathbf{H}_{ji}^{s}=\mathbf{0}$,
$\forall j,i\in\Omega_{j}^{II}$ ).
\item \textbf{Utilization of Partial CSI from Feedback Strategy III}: From
feedback strategy III, we obtain the following set of spaces $\left\{ \mathbb{N}^{t}\left((\mathbf{S}_{j}^{r})^{\dagger}\mathbf{H}_{ji}^{s}\right):\forall j,i\in\Omega_{j}^{III}\right\} $.
Based on this information, we obtain the set of matrices $\{\mathbf{S}_{i}^{t}\in\mathbb{C}^{N_{i}^{s}\times N_{i}^{e}}:\forall i\}$
, where $(\mathbf{S}_{i}^{t})^{\dagger}\mathbf{S}_{i}^{t}=\mathbf{I}_{N_{i}^{e}}$,
\begin{equation}
\textrm{ span }\left(\mathbf{S}_{i}^{t}\right)=\underset{\forall j:\, i\in\Omega_{j}^{III}}{\bigcap}\mathbb{N}^{t}\left((\mathbf{S}_{j}^{r})^{\dagger}\mathbf{H}_{ji}^{s}\right)\label{eq:s_t_result}
\end{equation}
and
\begin{equation}
N_{i}^{e}=N_{i}^{s}-\sum_{j:\, i\in\Omega_{j}^{III}}M_{j}^{e},\forall i.\label{eq:N_e_define}
\end{equation}
Hence, we can design the precoder of Tx $i$ in the space of $\textrm{span}(\mathbf{S}_{i}^{t})$,
and consequently, all the interference from Tx $i$ to Rx $j$, where
$i\in\Omega_{j}^{III}$, is eliminated (note $(\mathbf{S}_{j}^{r})^{\dagger}\mathbf{H}_{ji}^{s}\mathbf{S}_{i}^{t}=\mathbf{0}$,
$\forall j,i\in\Omega_{j}^{III}$).
\item \textbf{Utilization of Partial CSI from Feedback Strategy IV}: From
feedback strategy IV, we obtain the following set of spaces, i.e.,
$\left\{ \textrm{span}\left(\left[\begin{array}{ccc}
\cdots & (\mathbf{S}_{j}^{r})^{\dagger}\mathbf{H}_{ji}^{s} & \cdots\end{array}\right]_{i\in\Omega_{j}^{IV}}^{T}\right):\forall j\right\} $. Based on this information, we find the set of matrices $\{\mathbf{\tilde{H}}_{j}:\forall j\}$,
where $\textrm{span}(\mathbf{\tilde{H}}_{j}^{T})=\textrm{span}\left(\left[\begin{array}{ccc}
\cdots & (\mathbf{S}_{j}^{r})^{\dagger}\mathbf{H}_{ji}^{s} & \cdots\end{array}\right]_{i\in\Omega_{j}^{IV}}^{T}\right)$, and there must exist an invertible matrix $\mathbf{R}_{j}\in\mathbb{C}^{M_{j}^{e}\times M_{j}^{e}}$
such that 
\begin{equation}
\mathbf{\tilde{H}}_{j}=\mathbf{R}_{j}\left[\begin{array}{ccc}
\cdots & (\mathbf{S}_{j}^{r})^{\dagger}\mathbf{H}_{ji}^{s} & \cdots\end{array}\right]_{i\in\Omega_{j}^{IV}}.\label{eq:H_hat_information}
\end{equation}
 Hence, we can obtain $\left\{ \mathbf{R}_{j}(\mathbf{S}_{j}^{r})^{\dagger}\mathbf{H}_{ji}^{s}:\forall j,i\in\Omega_{j}^{IV}\right\} $,
and the precoders / decorrelators can be designed based on these effective
CSI matrices such that the interference from Tx $p$ to Rx $q$, $\forall(p,q)\in\{(j,i):\forall j,i\in\Omega_{j}^{IV}\}$
can be aligned into a lower dimensional subspace at the Rxs.
\end{itemize}

Based on these insights, we propose the following structures for $\{\mathbf{V}_{i},\mathbf{U}_{j}\}$
in the MIMO interference networks.
\begin{definitn}
[IA Precoders / Decorrelators Structure]The IA solutions $\{\mathbf{V}_{i},\mathbf{U}_{j}\}$
for Problem \ref{IA-design-based-general} have the following structure:
\begin{equation}
\mathbf{V}_{i}=\left[\begin{array}{c}
\mathbf{S}_{i}^{t}\mathbf{V}_{i}^{a}\\
\mathbf{0}
\end{array}\right],\quad\mathbf{U}_{j}=\left[\begin{array}{c}
\mathbf{S}_{j}^{r}(\mathbf{R}_{j})^{\dagger}\mathbf{U}_{j}^{a}\\
\mathbf{0}
\end{array}\right],\label{eq:transceiver_structure}
\end{equation}
where $\mathbf{V}_{i}^{a}\in\mathbb{C}^{N_{i}^{e}\times d_{i}}$,
$\mathbf{U}_{j}^{a}\in\mathbb{C}^{M_{j}^{e}\times d_{j}}$, $M_{j}^{e}$
and $N_{i}^{e}$ are given in (\ref{eq:notaion_M_j^e}) and (\ref{eq:N_e_define})
respectively. \hfill \IEEEQED
\end{definitn}

Note the above solution structures (\ref{eq:transceiver_structure})
automatically satisfy the IA constraints (\ref{eq:IA_condition})
for links from Tx $i$ to Rx $j$, where $i\in\Omega_{j}^{II}\Omega_{j}^{III}$,
$\forall j$ and they satisfy the partial CSI feedback constraints
in Problem \ref{IA-design-based-general}. However, the constraints
that $\{\mathbf{H}_{ji}:\forall j,i\in i\in\Omega_{j}^{I}\}$ are
not fed back and hence can not be utilized to design the precoders,
still make it hard to apply classical \emph{Algebraic Geometry} theory
\cite{razaviyayn2011degrees,ruan2012feasibility} to the study of
Problem 1. To cope with this, we further transform Problem 1 to the
following feasibility problem in which all the hidden constraints
on the available CSI knowledge are explicitly handled. 
\begin{problem}
[Transformed IA Problem]\label{CIA-under-transceiver}Find $\mathbf{V}_{i}^{a}\in\mathbb{C}^{N_{i}^{e}\times d_{i}}$,
$\forall i$ and $\mathbf{U}_{j}^{b}\in\mathbb{C}^{M_{j}^{e}\times d_{j}^{0}}$,
$\forall j$ such that $\{\mathbf{V}_{i}^{a},\mathbf{U}_{j}^{b}\}$
satisfy
\begin{equation}
\textrm{rank}(\mathbf{V}_{i}^{a})=d_{i},\;\forall i,\,\textrm{rank}(\mathbf{U}_{j}^{b})=d_{j}^{0},\;\forall j,\label{eq:problem_2_rank_conditioin}
\end{equation}
\begin{equation}
(\mathbf{U}_{j}^{b})^{\dagger}\mathbf{G}_{ji}\mathbf{V}_{i}^{a}=\mathbf{0},\,\forall j,i\in\Omega^{IV},\label{eq:equivalent_problem2}
\end{equation}
where $\mathbf{G}_{ji}=\mathbf{R}_{j}(\mathbf{S}_{j}^{r})^{\dagger}\mathbf{H}_{ji}^{s}\mathbf{S}_{i}^{t}$,
$d_{j}^{0}=d_{j}+\sum_{i\in\Omega_{j}^{I}}d_{i}$, $\forall j,i\in\Omega^{IV}$.\hfill \IEEEQED
\end{problem}

\begin{lemma}
[Equivalence of Problem 1 and Problem \ref{CIA-under-transceiver}]\label{With-the-transceiver}For
a given feedback profile $\mathcal{L}$, under the precoders / decorrelators
structures in (\ref{eq:transceiver_structure}), we have almost surely
that Problem 1 is feasible iff Problem 2 is feasible. Furthermore,
if $\{\mathbf{V}_{i}^{a},\mathbf{U}_{j}^{b}\}$ are the solution of
Problem \ref{CIA-under-transceiver}, then 
\begin{equation}
\mathbf{V}_{i}=\left[\begin{array}{c}
\mathbf{S}_{i}^{t}\mathbf{V}_{i}^{a}\\
\mathbf{0}
\end{array}\right],\mathbf{U}_{j}=v_{d_{j}}\left(\sum_{i\neq j}\left(\mathbf{H}_{ji}\mathbf{V}_{i}\right)\left(\mathbf{H}_{ji}\mathbf{V}_{i}\right)^{\dagger}\right),\forall i,j\label{eq:solution_calculation}
\end{equation}
are solution of Problem \ref{IA-design-based-general}, where $v_{(d)}(\mathbf{A})$
is the matrix of eigenvectors  corresponding to the $d$ least eigenvalues
of a Hermitian matrix $\mathbf{A}$.\end{lemma}
\begin{proof}
Please See Appendix \ref{sub:Proof-of-reduced}.
\end{proof}

Note that Lemma \ref{With-the-transceiver} simplifies the original
Problem 1 by eliminating the IA constraints on all links from Tx $i$
to Rx $j$, where $i\in\Omega_{j}^{I}\cup\Omega_{j}^{II}\Omega_{j}^{III}$,
$\forall j$. Furthermore, the solutions obtained will automatically
satisfy the partial CSI knowledge constraints.

\subsection{IA Precoders / Decorrelators Design }

In this section, we will derive solutions for Problem 1 by solving
Problem 2. Note that Problem 2 is a bi-convex problem w.r.t. $\{\mathbf{V}_{i}^{a}\}$
and $\{\mathbf{U}_{j}^{b}\}$. As a result, we shall apply \emph{alternating
optimization} techniques \cite{gomadam2011distributed,santamaria2010maximum}
to obtain a local optimal solution. The algorithm details are outlined
below:
\begin{equation}
\min_{\{\mathbf{U}_{j}^{b}\in\mathbb{C}^{M_{j}^{e}\times d_{j}^{0}}:(\mathbf{U}_{j}^{b})^{\dagger}\mathbf{U}_{j}^{b}=\mathbf{I}_{d_{j}^{0}},\forall j\}}I=\underset{j,i:\, i\in\Omega_{j}^{IV}}{\sum}\textrm{tr}\left(\left(\mathbf{U}_{j}^{b}\mathbf{G}_{ji}\mathbf{V}_{i}^{a}\right)\left(\mathbf{U}_{j}^{b}\mathbf{G}_{ji}\mathbf{V}_{i}^{a}\right)^{\dagger}\right).\label{eq:solve_u}
\end{equation}
\begin{equation}
\min_{\{\mathbf{V}_{i}^{a}\in\mathbb{C}^{N_{i}^{e}\times d_{i}}:(\mathbf{V}_{j}^{a})^{\dagger}\mathbf{V}_{j}^{a}=\mathbf{I}_{d_{j}},\forall i\}}I=\underset{j,i:\, i\in\Omega_{j}^{IV}}{\sum}\textrm{tr}\left(\left(\mathbf{U}_{j}^{b}\mathbf{G}_{ji}\mathbf{V}_{i}^{a}\right)\left(\mathbf{U}_{j}^{b}\mathbf{G}_{ji}\mathbf{V}_{i}^{a}\right)^{\dagger}\right).\label{eq:solve_v}
\end{equation}

\textit{Algorithm 1 }\textit{\emph{(}}\emph{Iterative Precoder / Decorrelator
Design}\textit{):}
\begin{itemize}
\item \textbf{Step 1} \textbf{(}\emph{Initialization}\textbf{)}: Randomly
initialize $\mathbf{V}_{i}^{a}\in\mathbb{C}^{N_{i}^{e}\times d_{i}}$,
$\mathbf{U}_{i}^{b}\in\mathbb{\mathbb{C}}^{M_{i}^{e}\times d_{i}^{0}}$,
$\forall i$. Initialize $\mathbf{G}_{ji}=\mathbf{R}_{j}(\mathbf{S}_{j}^{r})^{\dagger}\mathbf{H}_{ji}^{s}\mathbf{S}_{i}^{t}$,
$\forall j,i\in\Omega_{j}^{IV}$, where $\mathbf{S}_{i}^{t}$ is given
in (\ref{eq:s_t_result}).
\item \textbf{Step 2 (}\emph{Update}\textbf{ $\{\mathbf{U}_{j}^{b}\}$ }\emph{by
Solving (\ref{eq:solve_u})}\textbf{)}: Update $\mathbf{U}_{j}^{b}=v_{d_{j}^{0}}(\mathbf{E}_{j})$
where $\mathbf{E}_{j}=\underset{i:i\in\Omega_{j}^{IV}}{\sum}\left(\mathbf{G}_{ji}\mathbf{V}_{i}^{a}\right)\left(\mathbf{G}_{ji}\mathbf{V}_{i}^{a}\right)^{\dagger}$,
$\forall j$.
\item \textbf{Step 3} \textbf{(}\emph{Update}\textbf{ $\{\mathbf{V}_{i}^{a}\}$}\emph{
by Solving (\ref{eq:solve_v})}\textbf{)}: Update $\mathbf{V}_{i}^{a}=v_{d_{i}}(\mathbf{T}_{i})$,
where $\mathbf{T}_{i}=\underset{j:\, i\in\Omega_{j}^{IV}}{\sum}\left(\mathbf{G}_{ji}^{\dagger}\mathbf{U}_{j}^{b}\right)\left(\mathbf{G}_{ji}^{\dagger}\mathbf{U}_{j}^{b}\right)^{\dagger}$,
$\forall i$.
\item Repeat \textbf{Step 2} and \textbf{Step 3} until convergence. From
the converged solution of $\{\mathbf{V}_{i}^{a},\mathbf{U}_{j}^{b}\}$
above, we can get the overall solution $\{\mathbf{V}_{i},\mathbf{U}_{j}\}$
of Problem 1 using (\ref{eq:solution_calculation}). \hfill \IEEEQED
\end{itemize}

\mysubnote{the convergence of the Algorithm 1}
\begin{remrk}
[Characterization of Algorithm 1]Note Algorithm 1 can automatically
adapt to the \emph{partial CSI feedback} constraint in Problem \ref{IA-design-based-general}
for a given $\mathcal{L}$. On the other hand, Algorithm 1 converges
almost surely because the total interference leakage $I$ in (\ref{eq:solve_u})
and (\ref{eq:solve_v}) is non-negative and it is monotonically decreasing
in the alternating updates of Step 2 and Step 3. Note that if the
total interference leakage $I$ at the converged local optimal point
is 0, then the solution is a feasible solution of Problem 1. 
\end{remrk}

\section{Feasibility Conditions and Feedback Profile Design}

In this section, we study the feasibility conditions of Problem \ref{IA-design-based-general}
under a feedback profile $\mathcal{\mathcal{L}}$ and the precoder
/ decorrelator structure in (\ref{eq:transceiver_structure}). Based
on the feasibility conditions, a low complexity greedy algorithm is
further proposed to derive a feedback profile $\mathcal{L}$ for a
\emph{given} DoF requirements in the interference network. The derived
feedback profile can achieve substantial savings in the total CSI
feedback dimension required to achieve the given DoFs.

\subsection{Feasibility Conditions under Feedback Profile $\mathcal{L}$}

In this section, we extend the results in \emph{Algebraic Geometry}
\cite{razaviyayn2011degrees,ruan2012feasibility} and establish new
feasibility conditions for IA under reduced CSI feedback dimension.
We first have the following property regarding Problem \ref{CIA-under-transceiver}.

\mysubnote{formally give the TIP property}
\begin{lemma}
[Transformation Invariant Property]\label{Rotational-Invariant-PropertyThe}The
invertible matrices $\{\mathbf{R}_{j}\}$ do not affect the feasibility
conditions of Problem \ref{CIA-under-transceiver}, i.e., Problem
\ref{CIA-under-transceiver} is feasible when $\mathbf{R}_{j}=\mathbf{I}$,
$\forall j$ iff it is feasible under any invertible matrix $\mathbf{R}_{j}$,
$\forall j$.\end{lemma}
\begin{proof}
Please see Appendix \ref{sub:Proof-of-Lemma-RIP} for details.\end{proof}
\begin{remrk}
[Role of Lemma \ref{Rotational-Invariant-PropertyThe}]From Lemma
\ref{Rotational-Invariant-PropertyThe}, we further conclude that
it is sufficient to feedback the row space of the concatenated CSI
matrices, i.e., $\textrm{span}\left(\left[\begin{array}{ccc}
\cdots & (\mathbf{S}_{j}^{r})^{\dagger}\mathbf{H}_{ji}^{s} & \cdots\end{array}\right]_{i\in\Omega_{j}^{IV}}^{T}\right)$ at each Rx in order to satisfy the IA constraints in (\ref{eq:problem_2_rank_conditioin})-(\ref{eq:equivalent_problem2}).
This is illustrated in the IV-th feedback strategy in our proposed
feedback structure in (\ref{eq:s_j^r_infor}). In general, the feedback
dimension will be reduced by adopting feedback strategy IV, while
the feasibility of Problem \ref{CIA-under-transceiver} is not affected
(i.e., the same as feeding back $\left[\begin{array}{ccc}
\cdots & (\mathbf{S}_{j}^{r})^{\dagger}\mathbf{H}_{ji}^{s} & \cdots\end{array}\right]_{i\in\Omega_{j}^{IV}}$). 
\end{remrk}

Since $\{\mathbf{R}_{j}\}$ do not affect the problem feasibility,
we investigate the feasibility conditions under $\mathbf{R}_{j}=\mathbf{I}$,
$\forall j$ without loss of generality. The necessary feasibility
conditions are established as follows.

\mysubnote{necessary condition}
\begin{thm}
[Necessary Feasibility Conditions]\label{Feasibility-conditions}Given
a feedback profile $\mathcal{L}$ and the precoder / decorrelator
structure in (\ref{eq:transceiver_structure}), if Problem \ref{IA-design-based-general}
is feasible, then the following three conditions must be satisfied:
1) $N_{i}^{e}\geq d_{i}$, $\forall i$; 2) $M_{j}^{e}\geq d_{j}^{0}$,
$\forall j$; 3) Denote $V_{i}=d_{i}(N_{i}^{e}-d_{i})$, $\forall i$,
$U_{j}=d_{j}^{0}(M_{j}^{e}-d_{j}^{0})$, $\forall j$; $C_{ji}=d_{j}^{0}d_{i}$,
and $V_{i}$, $U_{j}$ and $C_{ji}$ satisfy
\begin{equation}
\sum_{j:\,(j,i)\in\Omega_{sub}}U_{i}+\sum_{i:\,(j,i)\in\Omega_{sub}}V_{i}\geq\sum_{j,i:\,(j,i)\in\Omega_{sub}}C_{ji},\quad\forall\Omega_{sub}\subseteq\{(j,i):\,\forall j,i\in\Omega_{j}^{IV}\}.\label{eq:eq:nece_condition-1}
\end{equation}
\end{thm}
\begin{proof}
Please see Appendix \ref{sub:Proof-of-Theorem_1} for details.
\end{proof}

Next, we try to study the sufficient feasibility conditions for Problem
1. To ensure that $\textrm{rank}(\mathbf{V}_{i}^{a})=d_{i}$ and $\textrm{rank}(\mathbf{U}_{j}^{b})=d_{j}^{0}$
in Problem \ref{CIA-under-transceiver}, it is sufficient to assume
that the first $d_{i}\times d_{i}$, $d_{j}^{0}\times d_{j}^{0}$
submatrix of $\mathbf{V}_{i}^{a}$, $\mathbf{U}_{j}^{b}$, denoted
by $\mathbf{V}_{i}^{(1)}$, $\mathbf{U}_{j}^{(1)}$, are invertible
$\forall i,\: j$. Under this assumption, we further denote $\mathbf{\tilde{V}}_{i}\in\mathbb{C}^{(N_{i}^{e}-d_{i})\times d_{i}}$,
$\mathbf{\tilde{U}}_{j}\in\mathbb{C}^{(M_{j}^{e}-d_{j}^{0})\times d_{j}^{0}}$,
and the four submatrices of $\mathbf{G}_{ji}$ in (\ref{eq:equivalent_problem2}),
i.e., $\mathbf{G}_{ji}^{(1)}\in\mathbb{C}^{d_{j}^{0}\times d_{i}}$,
$\mathbf{G}_{ji}^{(2)}\in\mathbb{C}^{(M_{j}^{e}-d_{j}^{0})\times d_{i}}$,
$\mathbf{G}_{ji}^{(3)}\in\mathbb{C}^{d_{j}^{0}\times(N_{i}^{e}-d_{i})}$,
$\mathbf{G}_{ji}^{(4)}\in\mathbb{C}^{(M_{j}^{e}-d_{j}^{0})\times(N_{i}^{e}-d_{i})}$,
as follows:
\[
\left[\begin{array}{c}
\mathbf{I}_{d_{i}}\\
\mathbf{\tilde{V}}_{i}
\end{array}\right]=\mathbf{V}_{i}^{a}\left(\mathbf{V}_{i}^{(1)}\right)^{-1},\left[\begin{array}{c}
\mathbf{I}_{d_{j}^{0}}\\
\mathbf{\tilde{U}}_{j}
\end{array}\right]=\mathbf{U}_{j}^{b}\left(\mathbf{U}_{j}^{(1)}\right)^{-1},\mathbf{G}_{ji}=\left[\begin{array}{cc}
\mathbf{G}_{ji}^{(1)} & \mathbf{G}_{ji}^{(3)}\\
\mathbf{G}_{ji}^{(2)} & \mathbf{G}_{ji}^{(4)}
\end{array}\right].
\]
Hence, equation (\ref{eq:equivalent_problem2}) becomes
\begin{equation}
\mathbf{G}_{ji}^{(1)}+\mathbf{\tilde{U}}_{j}^{\dagger}\mathbf{G}_{ji}^{(2)}+\mathbf{G}_{ji}^{(3)}\mathbf{\tilde{V}}_{i}+\mathbf{\tilde{U}}_{j}^{\dagger}\mathbf{G}_{ji}^{(4)}\mathbf{\tilde{V}}_{i}=\mathbf{0},\;\forall j,i\in\Omega_{j}^{IV}.\label{eq:polynomial_solution}
\end{equation}

Based on the equation sets in (\ref{eq:polynomial_solution}), the
sufficient feasibility conditions are established as follows.

\mysubnote{sufficient condition}
\begin{thm}
[Sufficient Feasibility Conditions]\label{Sufficient-Feasibility-Condition}Given
a feedback profile $\mathcal{L}$ and the precoder / decorrelator
structure structure in (\ref{eq:transceiver_structure}), if $N_{i}^{e}\geq d_{i}$,
$M_{i}^{e}\geq d_{i}^{0}$ $\forall i$, and the row vectors of all
the matrices $\{\mathbf{X}_{ji}:\forall j,i\in\Omega_{j}^{IV}\}$
are linearly independent, then Problem \ref{IA-design-based-general}
is feasible  almost surely, where 
\begin{equation}
\begin{array}{ccccccc}
\underset{d_{i}d_{j}^{0}\times\overline{M}}{\underbrace{\mathbf{X}_{ji}}} & = & \underset{d_{i}d_{j}^{0}\times m_{ji}}{[\underbrace{\mathbf{0}}} & (\mathbf{G}_{ji}^{(2)})^{T}\otimes\mathbf{I}_{d_{j}^{0}} & \underset{d_{i}d_{j}^{0}\times n_{ji}}{\underbrace{\mathbf{0}}} & \mathbf{I}_{d_{i}}\otimes\mathbf{G}_{ji}^{(3)} & \underset{d_{i}d_{j}^{0}\times k_{ji}}{\underbrace{\mathbf{0}}]}\end{array}\label{eq:X_determinant}
\end{equation}
and $\overline{M}=\sum_{i=1}^{K}\left(d_{i}^{0}(M_{i}^{e}-d_{i}^{0})+d_{i}(N_{i}^{e}-d_{i})\right)$,
$m_{ji}=\sum_{p=1}^{j-1}d_{p}^{0}(M_{p}^{e}-d_{p}^{0})$, $n_{ji}=\sum_{p=j+1}^{K}d_{p}^{0}(M_{p}^{e}-d_{p}^{0})+\sum_{q=1}^{i-1}d_{q}(N_{q}^{e}-d_{q})$,
$k_{ji}=\sum_{q=i+1}^{K}d_{q}(N_{q}^{e}-d_{q})$.

Moreover, under a given feedback profile $\mathcal{L}$, if the matrices
$\{\mathbf{X}_{ji}\}$ under a random channel realization $\{\mathbf{H}_{ji}:\forall j,i\}$
have linearly independent row vectors, Problem 1 is feasible for all
channel realizations almost surely.\end{thm}
\begin{proof}
Please See Appendix \ref{sub:Proof-of-Theorem-sufficient}.
\end{proof}

\mysubnote{sufficient  necessary condition in the divisible case}
\begin{cor}
[Feasibility Conditions in Divisible Cases]\label{Feasibility-Conditions-divisible}When
$d_{i}=d$, $\forall i$, and all $M_{i}^{s}$, $N_{i}^{s}$ are divisible
by the data stream, i.e., $d\mid M_{i}^{s}$, $d\mid N_{i}^{s}$,
$\forall\, i$, the three conditions in Theorem \ref{Feasibility-conditions}
are also sufficient.\end{cor}
\begin{proof}
Please see Appendix \ref{sub:Proof-of-Corollary-divisible}.
\end{proof}

\begin{remrk}
[Backward Compatibility with Previous Results]If the row space of
the concatenated channel matrices of all cross links are fed back,
i.e., $M_{i}^{s}=M_{i}$, $N_{i}^{s}=N_{i}$, $\forall i$, $\Omega_{j}^{I}=\Omega_{j}^{II}=\Omega_{j}^{III}=\emptyset$,
$\Omega_{j}^{IV}=\{1,\cdots j-1,j+1,\cdots K\}$ $\forall j$, then
$M_{i}^{e}=M_{i}$, $d_{j}^{0}=0$, $N_{j}^{e}=N_{j}$, $\forall i,\, j$,
and Corollary \ref{Feasibility-Conditions-divisible} reduces to the
results (Theorem 2) in \cite{razaviyayn2011degrees}.
\end{remrk}

\subsection{CSI Feedback Profile Design $\mathcal{L}$}

In this section, we focus on the design of the feedback profile to
reduce the total CSI feedback cost (feedback dimension) required to
achieve a given DoF requirement of the $K$ data streams $\{d_{1},d_{2},\cdots d_{K}\}$
in the MIMO interference networks. Specifically, we would like to
find a feedback profile $\mathcal{L}$ that satisfies the following
constraints:

\mysubnote{minimize the Sum feedback cost under given antenna resource and DoF budget.}
\begin{problem}
[Feedback Profile Design $\mathcal{L}$]\label{Minimize-Feedback-Cost}
\begin{eqnarray}
 &  & \bigcup_{m\in\{I,II,III,IV\}}\Omega_{j}^{m}=\{1,\cdots K\}/\{j\},\,\forall j,\label{eq:Feedback_topology_structure_condition1}\\
 &  & \Omega_{j}^{m}\bigcap\Omega_{j}^{n}=\emptyset,\forall m\neq n,\,\forall j,\label{eq:feedback_topology_structure_condition2}\\
 &  & N_{i}^{e}\geq d_{i},\, M_{i}^{e}\geq d_{i}^{0},\,\forall i,\label{eq:feasibility_condtion1}\\
 &  & \{\mathbf{X}_{ji}\}\textrm{ have linearly independent row vectors,}\,\forall j,i\in\Omega_{j}^{IV}\label{eq:feasibility_condition2}
\end{eqnarray}
where $\{\mathbf{X}_{ji}\}$ are given in Theorem \ref{Sufficient-Feasibility-Condition}
and $N_{i}^{e}$, $M_{i}^{e}$ and $d_{i}^{0}$ are given in (\ref{eq:notaion_M_j^e}),
(\ref{eq:N_e_define}) and (\ref{eq:problem_2_rank_conditioin}) respectively.\hfill \IEEEQED

Note that constraints (\ref{eq:Feedback_topology_structure_condition1}),
(\ref{eq:feedback_topology_structure_condition2}) come from the feedback
profile structure in Assumption \ref{Structure-of-Feedback}, constraints
(\ref{eq:feasibility_condtion1}) and (\ref{eq:feasibility_condition2})
come from the the feasibility conditions of Problem \ref{IA-design-based-general}
(Theorem \ref{Sufficient-Feasibility-Condition}). 

A feedback profile that satisfies the above constraints is called
a \emph{feasible feedback profile}. Ideally, we would like to find
a feasible feedback profile $\mathcal{L}$ that induces a small feedback
dimension. However, the design of feedback profile $\mathcal{L}$
is highly non-trivial due to the combinatorial nature, and doing exhaustive
search has exponential complexity in $\mathcal{O}\left((N)^{2K}4^{K(K-1)}(KN)^{3}\right)$
(see equation (\ref{eq:original_complexity})). In the following,
we propose a low complexity greedy algorithm to derive a feasible
feedback profile $\mathcal{L}$. We show in Section V and VI that
the associated feedback cost is quite small compared with conventional
state-of-the-art baselines.
\end{problem}

The details of the greedy algorithm are summarized as follows:

\mysubnote{give the details of the greedy algorithm}

\textit{Algorithm 2 (Greedy Feedback Profile Design $\mathcal{L}$):}
\begin{itemize}
\item \textbf{Step 1} \textit{(Initialization and Antenna Pruning)}: Initialize
$t=1$, $\Omega_{j}^{I}=\Omega_{j}^{II}=\Omega_{j}^{III}=\emptyset$,
$\Omega_{j}^{IV}=\{1,\cdots j-1,j+1,\cdots K\}$, $\forall j$, $M_{i}^{s}=\min(M_{i},\sum_{i}d_{i})$,
$N_{i}^{s}=\min(N_{i},\sum_{i}d_{i})$, $\forall i$ in $\mathcal{L}(t)$.
\item \textbf{Step 2} \textit{(Priority Computation)}: Compute the \emph{priority
$p(s)$} of update strategy $s$ on current $\mathcal{L}(t)$ as 
\begin{equation}
p(s)=\left(I_{\underset{\Delta V(s)\leq0\}}{\{\Delta D(s)\geq0}}\Delta D(s)\left(-\Delta V(s)+1\right)\alpha+I_{\underset{\Delta V(s)>0\}}{\{\Delta D(s)\geq0}}\frac{\Delta D(s)}{\Delta V(s)}\right),\forall s\in\mathcal{P}\left(\mathcal{L}(t)\right)\label{eq:priority}
\end{equation}
where $\mathcal{P}\left(\mathcal{L}(t)\right)$ is the space of the
update strategies on current $\mathcal{L}(t)$ and is given by 
\begin{equation}
\mathcal{P}(\mathcal{L}(t))=\left\{ \left\{ S^{I}(j,i),S^{II}(j,i),S^{III}(j,i):\forall j,i\in\Omega_{j}^{IV}\right\} ,\{S^{IV}(i),S^{V}(i),S^{VI}(i):\forall i\}\right\} ,\label{eq:update_strategy_space}
\end{equation}
$S^{I}(\cdot)\cdots S^{VI}(\cdot)$ are different types of update
operations described in Table \ref{fig:Description-of-different}
(note that all these update strategies could potentially reduce the
feedback dimension); $I_{\{\cdot\}}$ denotes the indicator function
and $\Delta D(s)$ denotes the dimension reduction via $s$, i.e.,
\begin{equation}
\Delta D(s)=D\left(\mathcal{L}(t+1\mid s)\right)-D\left(\mathcal{L}(t)\right),\label{eq:dimension_reduction}
\end{equation}
where $D(\mathcal{L})$ is in (\ref{eq:sum_feedback_dimension_expression})
and $\mathcal{L}(t+1\mid s)$ is the feedback profile obtained by
updating $\mathcal{L}(t)$ with $s$; $\Delta V(s)$ is the consumed
free variables with strategy $s$, i.e., 
\begin{equation}
\Delta V(s)=V\left(\mathcal{L}(t+1\mid s)\right)-V\left(\mathcal{L}(t)\right),\label{eq:variable_reduce}
\end{equation}
where $V(\mathcal{L})=\sum_{j}U_{j}+\sum_{i}V_{i}-\sum_{j,i:\,\Omega_{j}^{IV}=1}C_{ji}$,
and $U_{j}$, $V_{i}$, $C_{ji}$ are in Theorem \ref{Feasibility-conditions};
$\alpha$ is chosen to be $\alpha=K(\sum d_{i})^{2}$.
\item \textbf{Step 3} (\textit{Priority Sorting}\emph{):} Sort $\left\{ p(s):\forall s\in\mathcal{P}\left(\mathcal{L}(t)\right)\right\} $
in descending order, i.e., $\mathcal{P}\left(\mathcal{L}(t)\right)\triangleq\{s_{1},\cdots s_{J}\}$,
$J=|\mathcal{P}\left(\mathcal{L}(t)\right)|$ and $p(s_{1})\geq p(s_{2})\geq\cdots\geq p(s_{J})$.
Initialize the index $k=1$.
\item \textbf{Step 4} \emph{(Greedy Update on $\mathcal{L}$):}

\begin{itemize}
\item \textit{A (Update Trial and Stopping Condition): }If $k\leq|\mathcal{P}\left(\mathcal{L}(t)\right)|$
and $p(s_{k})\geq0$, then choose $s=s_{k}$ and update $\mathcal{L}$
as: $\mathcal{L}(t)\overset{s}{\rightarrow}\mathcal{L}(t+1\mid s)$;
Else, exit the algorithm.
\item \textit{B (Feasibility Checking): }If $\mathcal{L}(t+1\mid s)$ is
feasible by Theorem \ref{Sufficient-Feasibility-Condition}, then
set $t=t+1$, $\mathcal{L}(t+1)=\mathcal{L}(t+1\mid s)$ and go to
\textbf{Step 2}; Else, set $k=k+1$ and go to \textbf{Step }4 \emph{A}.\hfill \IEEEQED
\end{itemize}
\end{itemize}

\begin{table}
\begin{centering}
\begin{tabular}{|c|c|}
\hline 
update strategy & update $\mathcal{L}(t)$ as\tabularnewline
\hline 
\hline 
$S^{I}(j,i)$ & $\Omega_{j}^{I}=\Omega_{j}^{I}\bigcup\{i\},\Omega_{j}^{IV}=\Omega_{j}^{IV}/\{i\}$\tabularnewline
\hline 
$S^{II}(j,i)$ & $\Omega_{j}^{II}=\Omega_{j}^{II}\bigcup\{i\},\Omega_{j}^{IV}=\Omega_{j}^{IV}/\{i\}$\tabularnewline
\hline 
$S^{III}(j,i)$ & $\Omega_{j}^{III}=\Omega_{j}^{III}\bigcup\{i\},\Omega_{j}^{IV}=\Omega_{j}^{IV}/\{i\}$\tabularnewline
\hline 
$S^{IV}(i)$ & $N_{i}^{s}=d_{i},\Omega_{j}^{IV}=\emptyset,\Omega_{j}^{II}=\Omega_{j}^{II}\bigcup\Omega_{j}^{IV}$\tabularnewline
\hline 
$S^{V}(i)$ & $M_{i}^{s}=M_{i}^{s}-1$\tabularnewline
\hline 
$S^{VI}(i)$ & $N_{i}^{s}=N_{i}^{s}-1$\tabularnewline
\hline 
\end{tabular}
\par\end{centering}

\caption{\label{fig:Description-of-different}Description of different update
strategies on $\mathcal{L}$}

\end{table}

\begin{remrk}
[Design Motivation of Algorithm 2]Given current feedback profile
$\mathcal{L}(t)$, different strategies in $\mathcal{P}(\mathcal{L}(t))$
in (\ref{eq:update_strategy_space}) have different features. For
instance, they reduce the feedback dimension differently (i.e., $\Delta D(s)$)
and consume different numbers of free variables (i.e., $\Delta V(s)$).
Intuitively, a strategy with a larger ratio of dimension reduction
versus variables consumption (i.e., $\frac{\Delta D(s)}{\Delta V(s)}$)
should have higher priority, as in this way, we may achieve more aggregate
feedback dimension reduction. On the other hand, strategies with $\Delta D(s)>0$,
$\Delta V(s)\leq0$ are given relatively higher priority, as illustrated
in (\ref{eq:priority}) (due to the factor $\alpha$ in (\ref{eq:priority})),
because these strategies reduce the feedback dimension (i.e., $\Delta D(s)>0$)
while they do not consume the free variables (i.e., $\Delta V(s)\leq0$). 
\end{remrk}

\mysubnote{analyze the complexity.}
\begin{remrk}
[Complexity of Greedy Feedback Profile Design]We compare the complexity
of exhaustive search and the proposed design algorithm as follows.
For simplicity, assume that $M_{i}=N_{i}=N$, $\forall i$. The overall
complexity of exhaustive search is 
\begin{equation}
\mathcal{O}\left(\underset{(c_{1})}{\underbrace{(N)^{2K}}}\underset{(c_{2})}{\underbrace{4^{K(K-1)}}}\underset{(c_{3})}{\underbrace{(KN)^{3}}}\right)\label{eq:original_complexity}
\end{equation}
where $(c_{1})$ is from the combinations of submatrix sizes, i.e.,
$M_{i}^{s},N_{i}^{s}\in\{1,\cdots\cdot,N\},\forall i$, $(c_{2})$
is from the combinations of cross link partitions, i.e., $i\in\Omega_{j}^{I}$,
$\Omega_{j}^{II}$, $\Omega_{j}^{III}$ or $\Omega_{j}^{IV}$, $\forall i\neq j$,
$\forall j$ and $(c_{3})$ is from the feasibility checking (See
Appendix \ref{sub:Complexity-of-Feasibility}). The overall worst-case
complexity of Algorithm 2 is 
\begin{equation}
\mathcal{O}\left(\underset{(c_{4})}{\underbrace{K^{2}}}\underset{(c_{5})}{\underbrace{\left(K^{2}+KN\right)}}\underset{(c_{6})}{\underbrace{(KN)^{3}}}\right)\label{eq:proposed_complexity}
\end{equation}
where $(c_{4})$ is from each update on $\mathcal{L}$ having at most
$3K(K-1)+3K$ trials ($|\mathcal{P}\left(\mathcal{L}(t)\right)|\leq3K(K-1)+3K$),
$(c_{5})$ is from there being less than $K(K-1)+2KN$ updates on
$\mathcal{L}$, and $(c_{6})$ is from the feasibility checking (See
Appendix \ref{sub:Complexity-of-Feasibility}).
\end{remrk}

\mysubnote{performance of the Algorithm under a special topology}

\section{Tradeoff Analysis of DoF and CSI Feedback dimension}

In this section, we analyze the tradeoff between the DoF, antenna
resource and feedback cost for MIMO interference network under the
proposed feedback profile design. To obtain some simple insights,
we shall give a closed form expression on the tradeoff for a symmetric
MIMO interference network.
\begin{thm}
[Performance-Cost Tradeoff on a Symmetric MIMO Interference Network]\label{Performance-on-special}Consider
a $K$-user MIMO interference network where $d_{i}=d$, $M_{i}=M$,
$N_{i}=N$, $\forall i$ and $M$, $N$, $d$ satisfy $2\mid M$,
$M\leq2K+1$, $N=\frac{1}{2}KM$, $d\mid M$. The tradeoff between
the data stream $d$ and the feedback dimension $D_{p}$ is summarized
below:

\vspace{2bp}

\begin{tabular}{|c|c||c|}
\hline 
Data Stream $d$ & Feedback dimension $D_{p}$ & Feedback Profile $\mathcal{L}$\tabularnewline
\hline 
\hline 
$d\leq\frac{M}{K},d\mid M$ & 0 & $\begin{array}{c}
N_{i}^{s}=d,\, M_{i}^{s}=M,\forall i;\\
\Omega_{j}^{I}=\{1,\cdots j-1,j+1,\cdots K\};\\
\Omega_{j}^{II}=\Omega_{j}^{III}=\Omega_{j}^{IV}=\emptyset,\forall j.
\end{array}$\tabularnewline
\hline 
$\begin{array}{c}
d=\frac{M}{K-\kappa},\\
1\leq\kappa\leq K-2
\end{array}$ & $\begin{array}{c}
\left((K+1)d^{2}-Md\right)\\
\cdot(K-1)^{2}
\end{array}$ & $\begin{array}{c}
N_{i}^{s}=Kd,\; M_{i}^{s}=M,\forall i\in\{1,\cdots\kappa+1\};\\
N_{i}^{s}=d,\; M_{i}^{s}=M-d,\forall i\in\{\kappa+2,\cdots K\};\\
\Omega_{j}^{I}=\Omega_{j}^{IV}=\emptyset;\\
\Omega_{j}^{II}=\{\kappa+2,\cdots K\}/\{j\};\\
\Omega_{j}^{III}=\{1\cdots\kappa+1\}/\{j\}.
\end{array}$\tabularnewline
\hline 
\end{tabular}

\vspace{3bp}
\end{thm}
\begin{proof}
See Appendix \ref{sub:Proof-of-Theorem-performance}.
\end{proof}

\mysubnote{Remark on the comparison}
\begin{remrk}
[Interpretation of Theorem \ref{Performance-on-special}]\label{Interprestation-of-Theorem}From
the tradeoff expression between the DoF, antenna resource, and feedback
cost in Theorem \ref{Performance-on-special}, we can obtain the following
insights:\end{remrk}
\begin{itemize}
\item \textbf{Feedback Dimension versus DoF }$d$: Since $D_{p}=\left((K+1)d^{2}-Md\right)(K-1)^{2}$,
$\frac{M}{K-1}\leq d\leq\frac{M}{2},d\mid M$, we observe that given
the number of antennas $M$, there is a quadratic increase of $D_{p}$
w.r.t. $d$. Hence the feedback cost tends to increase faster as $d$
becomes larger. 
\item \textbf{Feedback Dimension versus Number of Antennas} $M$: Since
 $D_{p}=\left((K+1)d^{2}-Md\right)(K-1)^{2}$, $\frac{M}{K-1}\leq d\leq\frac{M}{2},d\mid M$,
we observe that given a DoF requirement $d$, the feedback cost tends
to decrease as the number of antennas $M$ increases. This is because
as $M$ gets larger, we obtain larger freedom for the feedback profile
design, and hence a better feedback profile could be obtained.
\end{itemize}

We further compare the result derived in Theorem \ref{Performance-on-special}
with a common baseline, which feedbacks the full channel direction
of all the cross links in the symmetric MIMO interference network
\cite{krishnamachari2009interference,thukral2009interference}. In
this baseline, the feedback function is given by $F_{j}(\mathcal{H}_{j})=\left(\begin{array}{ccc}
\cdots, & \{a\mathbf{H}_{ji}:a\in\mathbb{C}\}, & \cdots\end{array}\right)_{\forall i\neq j}$, $\forall j$, and the feedback dimension is given by%
\footnote{Note that under full channel direction feedback, the maximum achievable
data stream $d$ is given by $d=\textrm{min}\left(\left\lfloor \frac{M+N}{K+1}\right\rfloor ,M,N\right)=\frac{M}{2}$
\cite{razaviyayn2011degrees,ruan2012feasibility}.%
} $D_{full}=K(K-1)(\frac{1}{2}KM^{2}-1)$, $1\leq d\leq\frac{M}{2}$,
$d\mid M$. Under the same DoF requirement $d$, the ratio of the
feedback dimension achieved by the proposed feedback profile and the
baseline is 
\[
\frac{D_{p}}{D_{full}}\leq\frac{d}{M},\,1\leq d\leq\frac{M}{2},d\mid M.
\]
Hence, the proposed feedback profile requires  a much lower feedback
dimension when $d\ll M$.

\section{Numerical Results}

In this section, we verify the performance of the proposed feedback-saving
scheme in MIMO interference networks through simulation. We consider
limited feedback with Grassmannian codebooks \cite{dai2008quantization}
to quantize the partial CSI $\{F_{j}\}$ at each Rx. The precoders
/ decorrelators are designed using the Algorithm 1 developed in Section
III-B. We consider $10^{4}$ i.i.d. Rayleigh fading channel realizations
and compare the performance of the proposed feedback scheme with the
following 3 baselines. 
\begin{itemize}
\item \textbf{Baseline 1} \emph{(Feedback Full CSI Direction }\cite{krishnamachari2009interference,thukral2009interference}\emph{)}:
Rxs quantize and feedback full CSI direction of all the cross links
using the Grassmannian codebooks \cite{krishnamachari2009interference,thukral2009interference}.
\item \textbf{Baseline 2}\emph{ (Feedback Truncated CSI} \cite{rezaee2012interference,rezaee2013csit}\emph{)}:
Rxs first truncate the part of the concatenated CSI that does not
affect classical IA feasibility \cite{rezaee2012interference,rezaee2013csit},
and then quantize and feedback the truncated CSI using the Grassmannian
codebooks. 
\item \textbf{Baseline 3} \emph{(Feedback Critical Amount of Truncated CSI)}:
Rxs first select the submatrices $\{\mathbf{H}_{ji}^{s}:\forall j,i,i\neq j\}$,
where $\mathbf{H}_{ji}^{s}=\left[\begin{array}{cc}
\mathbf{I}_{M_{j}^{s}} & \mathbf{0}\end{array}\right]\mathbf{H}_{ji}\left[\begin{array}{c}
\mathbf{I}_{N_{i}^{s}}\\
\mathbf{0}
\end{array}\right]$ and $\{M_{i}^{s},N_{i}^{s}\}$ are chosen to make the network tightly
IA feasible%
\footnote{Tightly IA feasible means that the IA feasible network would become
IA infeasible if we further reduce any of $\{M_{i}^{s},N_{i}^{s}\}$.%
}. Rxs then adopt the algorithm proposed in \cite{rezaee2012interference,rezaee2013csit}
to quantize and feedback the submatrices $\{\mathbf{H}_{ji}^{s}:\forall j,i,i\neq j\}$. 
\end{itemize}

In Fig. \ref{fig:DoF-versus-feedback} and Fig. \ref{fig:Feedback-dimensions-versus},
we consider a $K=4$, $[N_{1},\cdots N_{4}]=[5,4,4,3],$ $[M_{1},\cdots M_{4}]=[4,3,2,4]$,
$[d_{1},\cdots d_{4}]=[2,1,1,1]$ MIMO interference network. The obtained
sum feedback dimension for the proposed scheme, baseline 3, baseline
2 and baseline 1 are 38, 86, 111 and 144 respectively.

Fig. \ref{fig:DoF-versus-feedback} plots the network throughput versus
the sum limited feedback bits under transmit SNR 25 dB. The proposed
scheme outperforms all the baselines. This is because the proposed
scheme significantly reduces the CSI feedback dimension while preserving
the IA feasibility, and hence more feedback bits can be utilized to
reduce the quantization error per dimension. The dramatic performance
gain highlights the importance of optimizing the feedback dimension
in MIMO interference networks with limited feedback. 

\begin{figure}
\begin{centering}
\includegraphics[scale=0.6]{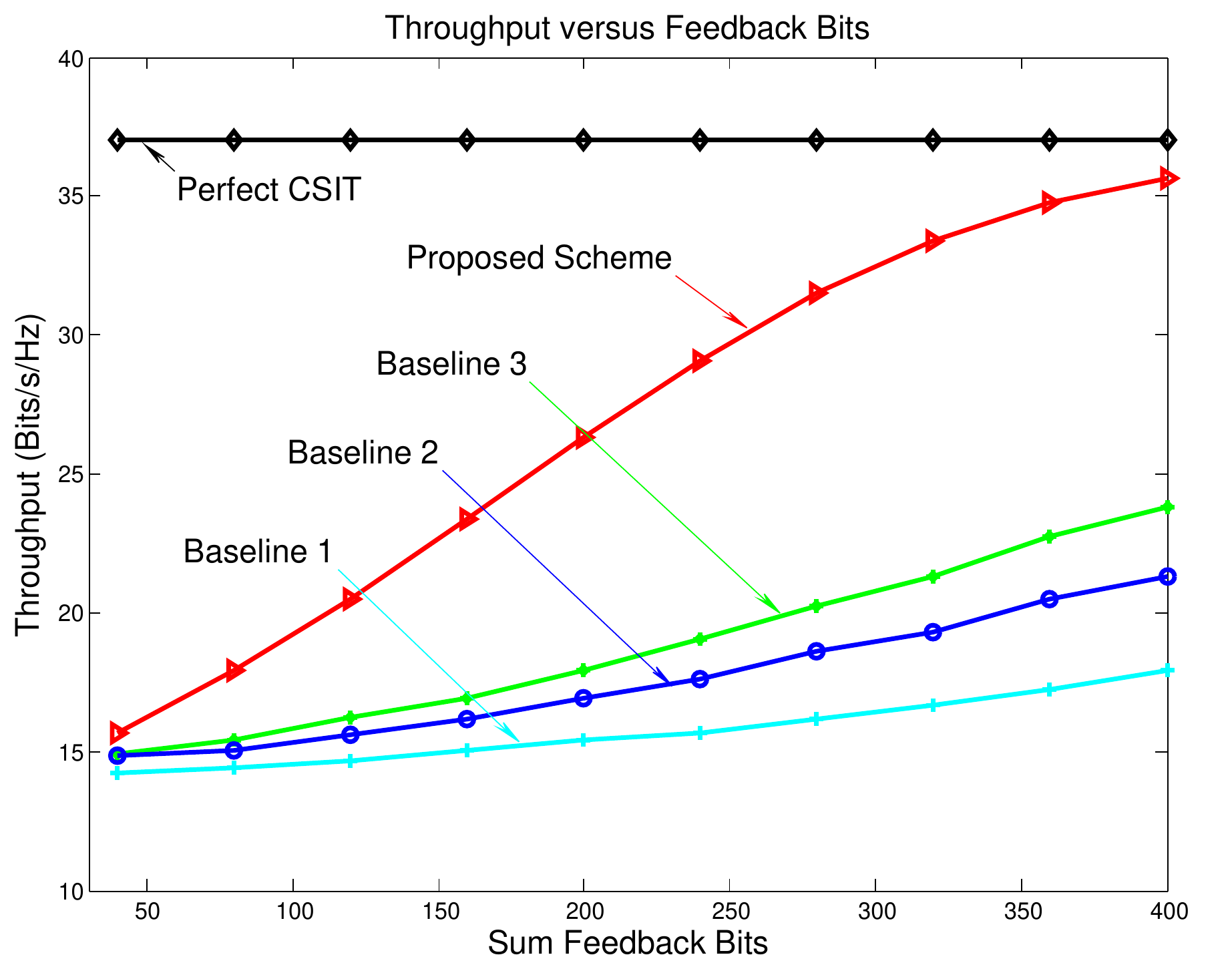}
\par\end{centering}

\caption{\label{fig:DoF-versus-feedback}Throughput versus sum feedback bits
under a $K=4$, $[N_{1},\cdots N_{4}]=[5,4,4,3],$ $[M_{1},\cdots M_{4}]=[4,3,2,4]$,
$[d_{1},\cdots d_{4}]=[2,1,1,1]$ MIMO interference network and the
average transmit SNR is 25 dB.}
\end{figure}

Fig. \ref{fig:Feedback-dimensions-versus} illustrates the network
throughput versus the transmit SNR under a total of $B=400$ feedback
bits. The proposed scheme achieves substantial throughput gain over
the baselines in a wide SNR region. The gain is larger at high SNR
because residual interference, which is the major performance bottleneck
in high SNR region, is significantly reduced by the proposed scheme. 

\begin{figure}
\begin{centering}
\includegraphics[scale=0.6]{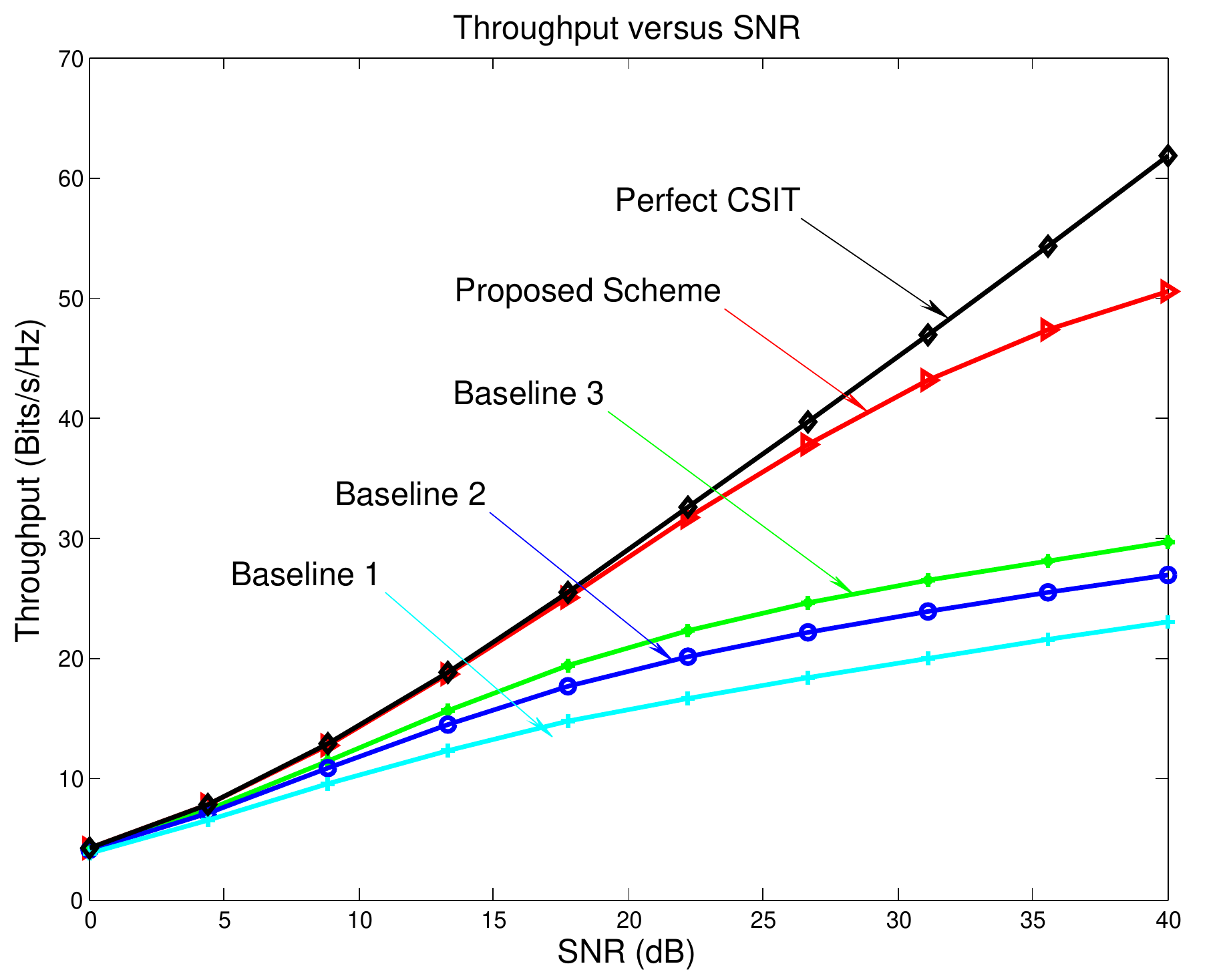}
\par\end{centering}

\caption{\label{fig:Feedback-dimensions-versus}Throughput versus SNR under
a $K=4$, $[N_{1},\cdots N_{4}]=[5,4,4,3],$ $[M_{1},\cdots M_{4}]=[4,3,2,4]$,
$[d_{1},\cdots d_{4}]=[2,1,1,1]$ MIMO interference network and the
sum feedback bit constraint is $400$.}
\end{figure}

\begin{table}
\begin{centering}
\begin{tabular}{|c|c|c|}
\hline 
Obtained Feedback Dimension & Fig. \ref{fig:DoF-versus-feedback}, \ref{fig:Feedback-dimensions-versus} & Fig. \ref{fig:Throughput-versus-sum-bits_new}, \ref{fig:Throughput-versus-SNR-new}\tabularnewline
\hline 
\hline 
Proposed Scheme & 38 & 20\tabularnewline
\hline 
Baseline 3 & 86 & 56\tabularnewline
\hline 
Baseline 2 & 111 & 72\tabularnewline
\hline 
Baseline 1 & 144 & 96\tabularnewline
\hline 
\end{tabular}
\par\end{centering}

\caption{Feedback dimension comparision}
\end{table}

In Fig. \ref{fig:Throughput-versus-sum-bits_new} and Fig. \ref{fig:Throughput-versus-SNR-new},
a $K=4$, $N_{i}=M_{i}=3,$ $d_{i}=1$, $\forall i$ MIMO interference
network is also simulated for performance comparison. The obtained
sum feedback dimension for the proposed scheme, baseline 3, baseline
2 and baseline 1 are 20, 56, 72 and 96 respectively. Fig. \ref{fig:Throughput-versus-sum-bits_new}
plots the network throughput versus the sum limited feedback bits
under transmit SNR 25 dB and Fig. \ref{fig:Throughput-versus-SNR-new}
illustrates the network throughput versus the transmit SNR under a
total of $B=200$ feedback bits. The proposed feedback scheme also
demonstrates significant performance advantages under this network
topology setup. 

\begin{figure}
\begin{centering}
\includegraphics[scale=0.6]{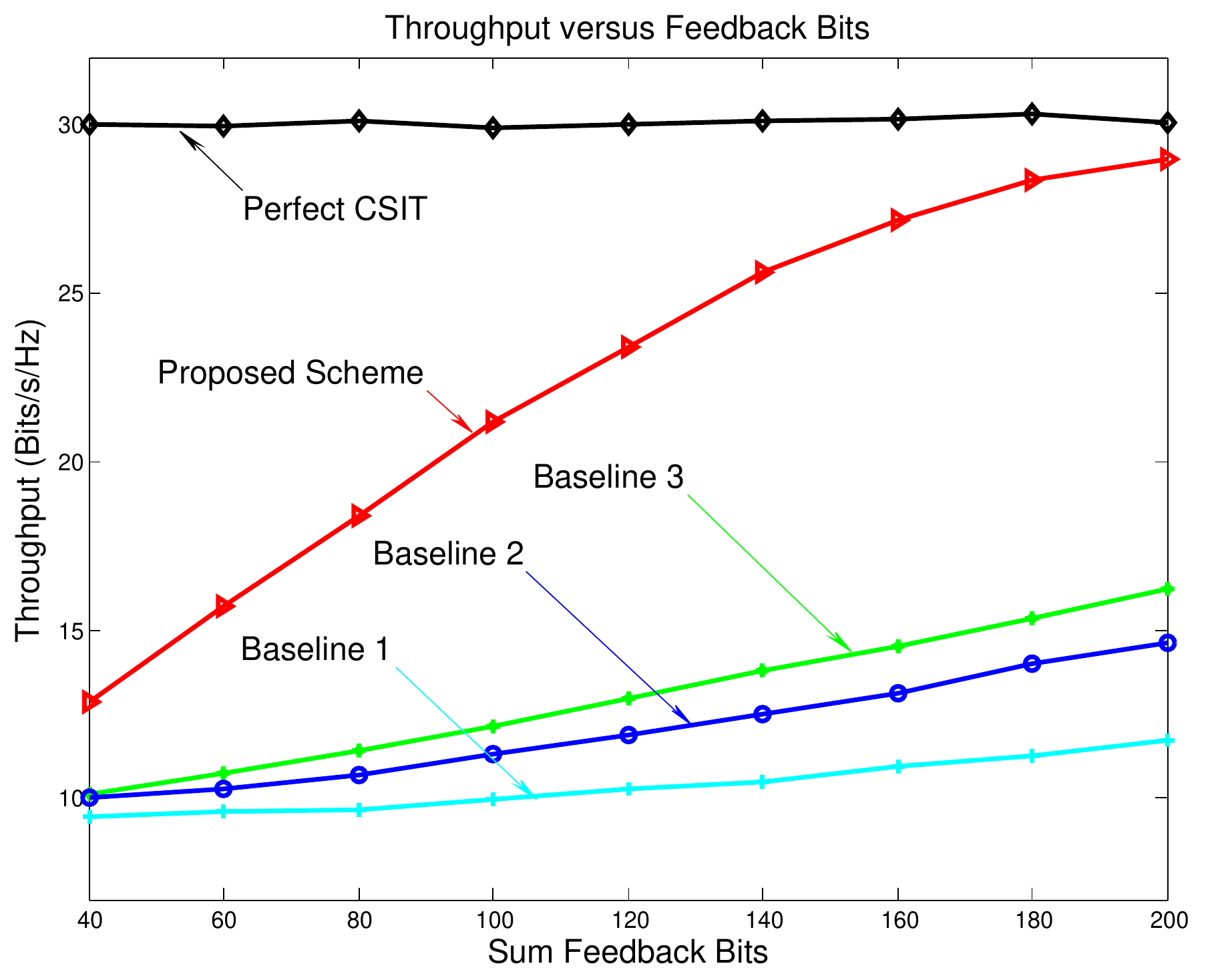}
\par\end{centering}

\caption{\label{fig:Throughput-versus-sum-bits_new}Throughput versus sum feedback
bits under a $K=4$, $N_{i}=M_{i}=3,$ $d_{i}=1$, $\forall i$ MIMO
interference network and the average transmit SNR is 25 dB.}
\end{figure}

\begin{figure}
\begin{centering}
\includegraphics[scale=0.6]{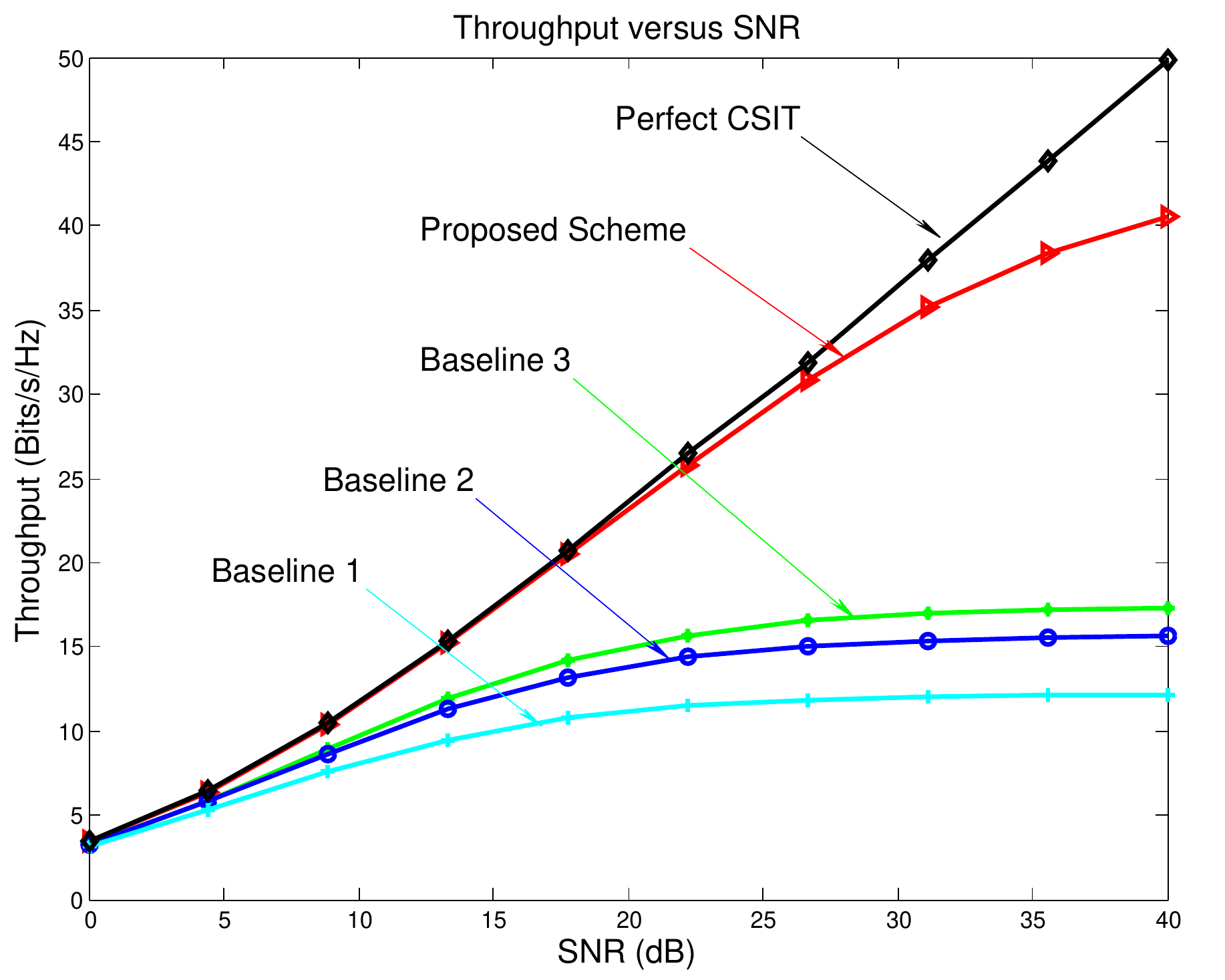}
\par\end{centering}

\caption{\label{fig:Throughput-versus-SNR-new}Throughput versus SNR under
a $K=4$, $N_{i}=M_{i}=3,$ $d_{i}=1$, $\forall i$ MIMO interference
network and the sum feedback bits constraint is $200$.}

\end{figure}

\section{Conclusions}

In this paper, we have proposed a low complexity IA design to achieve
a flexible tradeoff between the DoFs and CSI feedback cost. We characterize
the feedback cost by the feedback dimension. By exploiting the unique
features of IA algorithms, we propose a flexible feedback profile
design, which enables the Rxs to substantially reduce the feedback
cost by selecting the most critical part of CSI to feedback. We then
establish new feasibility conditions of IA under the proposed feedback
profile design. Finally, a low complexity algorithm for feedback profile
design is developed to reduce feedback dimension while preserving
IA feasibility. Both analytical and simulation results show that the
proposed scheme can significantly reduce the CSI feedback cost of
IA in MIMO interference networks. 

\appendix

\subsection{\label{sub:Proof-of-reduced}Proof of Lemma \ref{With-the-transceiver}}

By substituting the transceiver structure (\ref{eq:transceiver_structure})
into (\ref{eq:IA_condition}), we have that the constraints (\ref{eq:IA_condition})
in Problem 1 are satisfied for all links from Tx $p$ to Rx $q$,
where $(q,p)\in\{(j,i):\,\forall j,i\in\Omega_{j}^{II}\bigcup\Omega_{j}^{III}\}$.
Hence, the remaining constraints in Problem 1 are reduced to
\begin{equation}
\textrm{rank}(\mathbf{V}_{i}^{a})=d_{i},\quad\forall i,\textrm{rank}(\mathbf{U}_{j}^{a})=d_{j},\quad\forall j;\label{eq:sub_condition}
\end{equation}
\begin{equation}
(\mathbf{U}_{j}^{a})^{\dagger}\left[\begin{array}{ccc}
\cdots & \mathbf{G}_{ji}\mathbf{V}_{i}^{a} & \cdots\end{array}\right]_{i\in\Omega_{j}^{I}\bigcup\Omega_{j}^{IV}}=\mathbf{0},\quad\forall j.\label{eq:reduced_alignment_condition-1}
\end{equation}
Note that (a) $\mathbf{V}_{i}^{a}$, $\mathbf{U}_{i}^{a}$ are functions
of $\{\mathbf{H}_{ji}:\forall j,i\in\Omega_{j}^{I}\bigcup\Omega_{j}^{IV}\}$
and hence are independent of $\{\mathbf{H}_{ii}:\forall i\}$; and
(b) the entries of $\mathbf{H}_{ii}$ are i.i.d Gaussian distributed;
we have that (\ref{eq:MIA}) and (\ref{eq:sub_condition}) are equivalent
almost surely. Condition (\ref{eq:reduced_alignment_condition-1})
and $\textrm{rank}(\mathbf{U}_{j}^{a})=d_{j}$ in (\ref{eq:sub_condition})
are equivalent to 
\begin{equation}
\textrm{dim}_{s}\left(\textrm{span}\left(\left\{ \mathbf{G}_{ji}\mathbf{V}_{i}^{a}\right\} {}_{i\in\Omega_{j}^{I}\bigcup\Omega_{j}^{IV}}\right)\right)\leq M_{j}^{e}-d_{j}.\label{eq:condition2}
\end{equation}
Since links in $\{(j,i):i\in\Omega_{j}^{I}\}$ are not fed back, $\left\{ \mathbf{G}_{ji}\mathbf{V}_{i}^{a}\right\} {}_{i\in\Omega_{j}^{I}}$
will span a random subspace with dimension $(d_{j}^{0}-d_{j})$, which
is independent of $\textrm{span}\left(\left\{ \mathbf{G}_{ji}\mathbf{V}_{i}^{a}\right\} {}_{i\in\Omega_{j}^{IV}}\right)$\cite{feng2007rank}.
Hence, (\ref{eq:condition2}) can be equivalently transformed to 
\begin{equation}
\textrm{dim}_{s}\left(\textrm{span}\left(\left\{ \mathbf{G}_{ji}\mathbf{V}_{i}^{a}\right\} {}_{i\in\Omega_{j}^{IV}}\right)\right)\leq M_{j}^{e}-d_{j}-(d_{j}^{0}-d_{j}),\;\forall j.\label{eq:dimension_feasibility}
\end{equation}
\begin{equation}
\Longleftrightarrow\exists\mathbf{U}_{j}^{b}\in\mathbb{\mathbb{C}}^{M_{j}^{e}\times d_{j}^{0}},\;\textrm{rank}(\mathbf{U}_{j}^{b})=d_{j}^{0},\;(\mathbf{U}_{j}^{b})^{\dagger}\mathbf{G}_{ji}\mathbf{V}_{i}^{a}=\mathbf{0},\;\forall j,i\in\Omega_{j}^{IV}.\label{eq:equivalent}
\end{equation}
Hence, we prove the equivalence between Problem 1 and Problem \ref{CIA-under-transceiver}.

We derive the relationship between the solutions of the Problem 1
and Problem 2 as follows. Assume $\{\mathbf{V}_{i}^{a},\mathbf{U}_{j}^{b}\}$
are the solution of Problem \ref{CIA-under-transceiver}. Then there
exists $\{\mathbf{V}_{i}^{a},\mathbf{U}_{j}^{a}\}$ such that (\ref{eq:reduced_alignment_condition-1})
is satisfied and 
\[
\left[\begin{array}{c}
\mathbf{S}_{j}^{r}\mathbf{U}_{j}^{a}\\
\mathbf{0}
\end{array}\right]^{\dagger}\mathbf{H}_{ji}\left[\begin{array}{c}
\mathbf{S}_{i}^{t}\mathbf{V}_{i}^{a}\\
\mathbf{0}
\end{array}\right]=\mathbf{0},\forall i\Rightarrow
\]
 
\[
\textrm{dim}_{s}\left(\textrm{span}\left(\left\{ \mathbf{H}_{ji}\left[\begin{array}{c}
\mathbf{S}_{i}^{t}\mathbf{V}_{i}^{a}\\
\mathbf{0}
\end{array}\right]\right\} _{\forall i}\right)\right)\leq M_{j}-d_{j}.
\]
Hence, the least $d_{j}$ eigenvalues of the Hermitian matrix $\sum_{i\neq j}\left(\mathbf{H}_{ji}\mathbf{V}_{i}\right)\left(\mathbf{H}_{ji}\mathbf{V}_{i}\right)^{\dagger}$
are 0, and $\{\mathbf{V}_{i},\mathbf{U}_{j}\}$ given by (\ref{eq:solution_calculation})
are the solution of Problem \ref{IA-design-based-general} almost
surely.

\subsection{\label{sub:Proof-of-Lemma-RIP}Proof of Lemma \ref{Rotational-Invariant-PropertyThe}}

Assume that Problem \ref{CIA-under-transceiver} is feasible under
$\mathbf{R}_{j}=\mathbf{I}$, $\forall j$. Then there must exist
$\mathbf{U}_{j}^{b}$, $\mathbf{V}_{i}^{a}$ such that 
\[
\mathbf{U}_{j}^{b}(\mathbf{S}_{j}^{r})^{\dagger}\mathbf{H}_{ji}^{s}\mathbf{S}_{i}^{t}\mathbf{V}_{i}^{a}=\mathbf{0},\,\forall j,i\in\Omega_{j}^{IV}.
\]
Then, for any invertible $\{\mathbf{R}_{j}\}$, we have 
\begin{equation}
(\hat{\mathbf{U}}_{j}^{b})^{\dagger}\mathbf{R}_{j}(\mathbf{S}_{j}^{r})^{\dagger}\mathbf{H}_{ji}^{s}\mathbf{S}_{i}^{t}\mathbf{V}_{i}^{a}=\mathbf{0},\,\forall j,i\in\Omega_{j}^{IV},\label{eq:equiva}
\end{equation}
where $\hat{\mathbf{U}}_{j}^{b}=(\mathbf{R}_{j}^{-1})^{\dagger}\mathbf{U}_{j}^{b}$.
Equation (\ref{eq:equiva}) shows that the IA constraints (\ref{eq:equivalent_problem2})
are satisfied under $\hat{\mathbf{U}}_{j}^{b}$. Therefore, Problem
\ref{CIA-under-transceiver} is still feasible under other invertible
$\mathbf{R}_{j}$.

The converse statement is trivial, hence Lemma \ref{Rotational-Invariant-PropertyThe}
is proved.

\subsection{\label{sub:Proof-of-Theorem_1}Proof of Theorem \ref{Feasibility-conditions}}

The necessity of conditions 1, 2 is straight forward. We focus on
proving the necessity of condition 3.

In the equation sets (\ref{eq:equivalent_problem2}) in Problem \ref{CIA-under-transceiver},
there are $U_{j}=d_{j}^{0}(M_{j}^{e}-d_{j}^{0})$ free variables in
$\mathbf{U}_{j}^{b}\in\mathbb{\mathbb{C}}^{M_{j}^{e}\times d_{j}^{0}}$,
$\forall j$, $V_{i}=d_{i}(N_{i}^{e}-d_{i})$ free variables in $\mathbf{V}_{i}^{a}\in\mathbb{C}^{N_{i}^{e}\times d_{i}}$,
$\forall i$ and $C_{ji}=d_{j}^{0}d_{i}$ scalar constraints in matrix
equation $(\mathbf{U}_{j}^{b})^{\dagger}\mathbf{G}_{ji}\mathbf{V}_{i}^{a}=\mathbf{0}$,
$\forall j,i\in\Omega_{j}^{IV}$ \cite{yetis2010feasibility,razaviyayn2011degrees}.
By analyzing the algebraic dependency of the IA constraints \cite{razaviyayn2011degrees},
we have that, the number of constraints should be no more than the
number of free variables for any subset of IA constraints; hence,
\begin{equation}
\sum_{j:\,(j,i)\in\Omega_{sub}}U_{j}+\sum_{i:\,(j,i)\in\Omega_{sub}}V_{i}\geq\sum_{j,i:\,(j,i)\in\Omega_{sub}}C_{ji},\:\forall\Omega_{sub}\subseteq\{(j,i):\,\forall j,i\in\Omega_{j}^{IV}\},\label{eq:feasibility}
\end{equation}
which is condition 3 in Theorem \ref{Feasibility-conditions}.

\subsection{\label{sub:Proof-of-Theorem-sufficient}Proof of Theorem \ref{Sufficient-Feasibility-Condition}}

(\ref{eq:polynomial_solution}) can be rewritten as 
\begin{equation}
\mathbf{y}_{ji}\triangleq\textrm{vec}(\mathbf{G}_{ji}^{(1)})+\mathbf{X}_{ji}\mathbf{v}+\textrm{vec}\left(\mathbf{\tilde{U}}_{j}^{\dagger}\mathbf{G}_{ji}^{(4)}\mathbf{\tilde{V}}_{i}\right)=\mathbf{0},\forall j,i\in\Omega_{j}^{IV},\label{eq:polynomial_map}
\end{equation}
where $\mathbf{v}$ is given by 
\[
\mathbf{v}=\left[\begin{array}{cccccc}
\textrm{vec}(\mathbf{\tilde{U}}_{1}^{\dagger})^{T} & \cdots & \textrm{vec}(\mathbf{\tilde{U}}_{K}^{\dagger})^{T} & \textrm{vec}(\mathbf{\tilde{V}}_{1})^{T} & \cdots & \textrm{vec}(\mathbf{\tilde{V}}_{K})^{T}\end{array}\right]^{T}.
\]

Note that each element in $\mathbf{y}_{ji}$ is a polynomial function
of the elements in $\mathbf{v}$. From (\ref{eq:polynomial_map}),
we have that $\mathbf{X}_{ji}$ defined in (\ref{eq:X_determinant})
are the coefficient vectors of the linear terms in $\mathbf{y}$.
According to \cite{razaviyayn2011degrees} (proof of Theorem 2) and
\cite{ruan2012feasibility} (Lemma 3.1-3.3), when the row vectors
of $\{\mathbf{X}_{ji}\}$ are linearly independent, equation sets
(\ref{eq:polynomial_map}) have solutions and hence Problem \ref{IA-design-based-general}
is feasible.

Adopting an approach similar to that in the proof of Corollary 3.1
in \cite{ruan2012feasibility}, we can further prove that under a
given feedback profile $\mathcal{L}$, the row vectors $\{\mathbf{X}_{ji}\}$
are either always linearly dependent or independent almost surely
for all channel realizations. Hence, when $\{\mathbf{X}_{ji}\}$ has
linearly independent rows under one random channel realization, Problem
1 is almost surely feasible.

\subsection{\label{sub:Proof-of-Corollary-divisible}Proof of Corollary \ref{Feasibility-Conditions-divisible}}

We prove Corollary \ref{Feasibility-Conditions-divisible} via the
following two lemmas (Lemma \ref{Sufficient-CIA-feasibility} and
Lemma \ref{Existence-of-variables}).

\mysubnote{divide into two steps to prove the sufficient side}

\mysubsubnote{one is to give the lemma which states sufficiency, the other claims the existence of the binary variables. }
\begin{lemma}
[Sufficient Feasibility Conditions]\label{Sufficient-CIA-feasibility}If
there exists a set of binary variables $\left\{ b_{jipq}^{t},b_{jipq}^{r}\in\{0,1\}\right\} $,
$\forall(j,i,p,q)\in\overline{\Omega}$ that satisfy the following
constraints, Problem \ref{CIA-under-transceiver} is almost surely
feasible. 
\begin{equation}
b_{jipq}^{t}+b_{jipq}^{r}=1,\;\forall j,i,p,q,\label{eq:constraints}
\end{equation}
\begin{equation}
\sum_{\underset{(j,i,p,q)\in\overline{\Omega}}{(i,q):}}b_{jipq}^{r}\leq U_{jp},\;\forall j,p,\label{eq:receiver_side}
\end{equation}
\begin{equation}
\sum_{\underset{(j,i,p,q)\in\overline{\Omega}}{(j,p):}}b_{jipq}^{t}\leq V_{iq},\;\forall i,q,\label{eq:transmit_side}
\end{equation}
\begin{equation}
b_{jip1}^{t}=\cdots=b_{jipd}^{t},\;\forall j,i,p,\label{eq:zhengtai}
\end{equation}
where $\overline{\Omega}=\{(j,i,p,q):p\in\{1,\cdots d+d_{j}^{0}\},q\in\{1,\cdots d\},i\in\Omega_{j}^{IV},\forall j\}$,
, $U_{jp}=(M_{j}^{e}-d_{j}^{0})$, $V_{iq}=(N_{i}^{e}-d)$.\end{lemma}
\begin{proof}
(Outline) Assume there exist binary variables $\left\{ b_{jipq}^{t},b_{jipq}^{r}\right\} $
satisfying (\ref{eq:constraints})-(\ref{eq:zhengtai}). It can be
proved that the row vectors of $\{\mathbf{X}_{ji}\}$ defined in Theorem
\ref{Sufficient-Feasibility-Condition} are linearly independent almost
surely and the proof is similar to that of \cite{ruan2012feasibility}
(Appendix G). We omit the details due to page limit.
\end{proof}

\begin{lemma}
[Existence of the Variables $\{b_{jipq}^{t},b_{jipq}^{r}\}$ in Divisible
Cases]\label{Existence-of-variables}Assume that the three conditions
in Theorem \ref{Feasibility-conditions} are satisfied and $d_{i}=d$,
$d\mid M_{i}^{s}$, $d\mid N_{i}^{s}$, $\forall i$. Then there exist
such binary variables $\left\{ b_{jipq}^{t},b_{jipq}^{r}\right\} $
satisfying conditions (\ref{eq:constraints})-(\ref{eq:zhengtai}).
\end{lemma}

\emph{Proof}: Condition (\ref{eq:feasibility}) is equivalent to the
following \cite{ruan2012feasibility}:

\begin{equation}
\sum_{\underset{(j,i,p,q)\in\overline{\Omega}_{sub}}{(j,p):}}U_{jp}+\sum_{\underset{(j,i,p,q)\in\overline{\Omega}_{sub}}{(i,q):}}V_{iq}\geq|\overline{\Omega}_{sub}|,\;\forall\overline{\Omega}_{sub}\subseteq\overline{\Omega}.\label{eq:sub_system}
\end{equation}

Conditional on (\ref{eq:sub_system}), we will prove the existence
of binary variables $\{b_{jipq}^{t},b_{jipq}^{r}\}$ satisfying (\ref{eq:constraints})-(\ref{eq:zhengtai})
via a constructive method. Specifically, we construct $\{b_{jipq}^{t},b_{jipq}^{r}\}$
by transforming the equation sets (\ref{eq:sub_system}) to the well
known max-flow problem \cite{schrijver2002history}. 

We first introduce a little about the max-flow problem. Denote $\mathcal{N}=(\mathcal{V},\mathcal{E})$
as a directed graph where $\mathcal{V}$ is the set of nodes and $\mathcal{E}$
the edges, $s$$,t\in\mathcal{V}$ are the source and sink node respectively.
The capacity of an edge, denoted by $c(u,v)$, represents the maximum
amount of flow that can pass through an edge. The flow of an edge,
denoted by $f(u,v)$ should satisfy $0\leq f(u,v)\leq c(u,v)$, $\forall(u,v)\in\mathcal{E}$
and the conservation of flows, i.e., 
\begin{equation}
\sum_{u:\,(u,v)\in\mathcal{E}}f(u,v)=\sum_{k:\,(v,k)\in\mathcal{E}}f(v,k)\label{eq:property_flow}
\end{equation}
$\forall v\in\mathcal{V}/\{s,t\}$. The value of the sum flow is defined
by $f_{sum}=\sum_{v\in\mathcal{V}}f(s,v)$. By adopting this mathematical
framework, we have the following lemma which help us to construct
$\{b_{jipq}^{t},b_{jipq}^{r}\}$.
\begin{lemma}
[Max-flow Problem]\label{Equivalent-feasibility-problem}The max
sum flow $f_{sum}$ of the graph $\mathcal{N}$ constructed in Algorithm
3 is $f_{sum}=\sum_{(j,i,p,q)\in\overline{\Omega}}1$ under the constraint
$f(v_{i1},c_{jip1})=f(v_{i2},c_{jip2})\cdots=f(v_{id},c_{jipd})$,
$f(u_{jp},c_{jip1})=f(u_{jp},c_{jip2})\cdots=f(u_{jp},c_{jipd})$,
$\forall$ $(j,i,p,q)\in\overline{\Omega}$, where $f(x,y)$ denotes
the edge flow from vertex $x$ to vertex $y$ in $\mathcal{N}$:
\end{lemma}
\emph{Algorithm 3 (Max Flow Graph $\mathcal{N}=(\mathcal{V},\mathcal{E})$):}
\begin{itemize}
\item \textbf{Step 1}: The vertices are given by $\mathcal{V}=\left\{ s,t,\{u_{jp},v_{iq},c_{jipq}:\forall(j,i,p,q)\in\overline{\Omega}\}\right\} $
where $s$ and $t$ are the source and sink node respectively. 
\item \textbf{Step 2}: The edges are given by $\mathcal{E}=\left\{ (s,u_{jp}),(s,v_{iq}),(u_{jp},c_{jipq}),(v_{iq},c_{jipq}),(c_{jipq},t):\forall(j,i,p,q)\in\overline{\Omega}\right\} $. 
\item \textbf{Step 3}: Set the edge capacity $c(s,u_{jp})=U_{jp}$, $c(s,v_{iq})=V_{iq}$,
$c(u_{jp},c_{jipq})=U_{jp}$, $c(v_{iq},c_{jipq})=V_{iq}$, $c(c_{jipq},t)=1$,
$\forall(j,i,p,q)\in\overline{\Omega}$.\end{itemize}
\begin{proof}
Please see Appendix \ref{sub:Equivalence-of-the} for the proof.
\end{proof}

\begin{figure}
\begin{centering}
\includegraphics[scale=0.9]{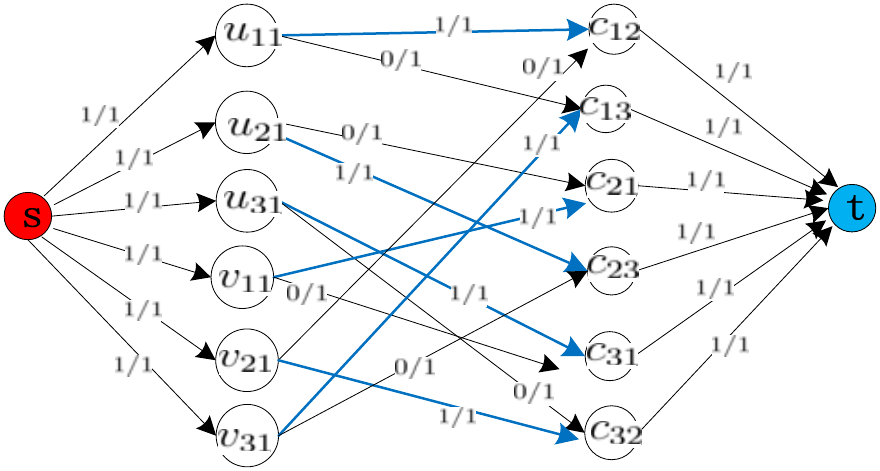}
\par\end{centering}

\caption{\label{fig:Max-flow-problem-example:}An example of constructed max-flow
graph for a $K=3$ user MIMO network under feedback profile parameters:
$N_{i}^{s}=M_{i}^{s}=2$, $d_{i}=1$, $\forall i$, $\Omega_{j}^{IV}=\{1,\cdots j-1,j+1,\cdots K\}$,
$\Omega_{j}^{I}=\Omega_{j}^{II}=\Omega_{j}^{III}=\emptyset$, $\forall j,i,j\neq i$.
The value $f/c$ near each edge denotes the flow ($f$) and capacity
($c$). }
\end{figure}

Fig. \ref{fig:Max-flow-problem-example:} illustrated an example of
constructed max-flow graph. Based on the flow values $\{f(u,v)\}$
in the flow graph $\mathcal{N}$, construct $b_{jipq}^{t},\, b_{jipq}^{r}$
as
\begin{equation}
b_{jipq}^{t}=f(v_{iq},c_{jipq}),\: b_{jipq}^{r}=f(u_{jp},c_{jipq}).\label{eq:binary_varialbes}
\end{equation}

From $f_{sum}=\sum_{(j,i,p,q)\in\overline{\Omega}}1$, we have $f(c_{jipq},t)=1$,
$\forall(j,i,p,q)\in\overline{\Omega}$. Note that $f(v_{iq},c_{jipq})$,
$f(u_{jp},c_{jipq})$ are integral as all capacity values on the edges
are integral \cite{cormen2001introduction}. Hence $b_{jipq}^{t}+b_{jipq}^{r}=f(c_{jipq},t)=1$
and $b_{jipq}^{t},b_{jipq}^{r}\in\{0,1\}$ according to (\ref{eq:property_flow}).
On the other hand, it is easy to verify that $\{b_{jipq}^{t},b_{jipq}^{r}\}$
satisfy the conditions (\ref{eq:receiver_side})-(\ref{eq:transmit_side})
as well according to (\ref{eq:property_flow}).\hfill \IEEEQED

\subsection{\label{sub:Equivalence-of-the}Proof of Lemma \ref{Equivalent-feasibility-problem}}

By the max-flow min-cut theorem \cite{cormen2001introduction}, the
max flow $f_{sum}\leq\sum_{(j,i,p,q)\in\overline{\Omega}}1=|\overline{\Omega}|$. 

We prove Lemma \ref{Equivalent-feasibility-problem} via the converse-negative
proposition. Assume that $f_{sum}<|\overline{\Omega}|$, then $\exists$
$(x,y,m,n)\in|\overline{\Omega}|$, such that $f(c_{xymn},t)=0$.
Due to the symmetry of the max-flow graph w.r.t. $q$, we must have
$f(c_{xym1},t)=\cdots=f(c_{xymd},t)=0$. Furthermore, the network
must have no further augmenting paths \cite{cormen2001introduction}
(otherwise, the max-flow can be increased). Construct $\overline{\Omega}_{sub}\subseteq\overline{\Omega}$
as follows:

\emph{Algorithm 4 (Construction of $\overline{\Omega}_{sub}$)}
\begin{itemize}
\item \textbf{Step 1}: Initialize $\mathcal{C}=\{c_{xym1},\cdots c_{xymd}\}$,
$\mathcal{C}_{c}=\{c_{jipq}:(j,i,p,q)\in\overline{\Omega}\}/\mathcal{C}$,
$\mathcal{U}=\{u_{xm},v_{y1},\cdots v_{yd}\}$ and $\mathcal{U}_{c}=\{u_{jp},v_{iq}:\forall(j,i,p,q)\in\overline{\Omega}\}/\mathcal{U}$
.
\item \textbf{Step 2}: For each $r\in\mathcal{C}_{c}$ such that $\exists$
$z\in\mathcal{U}$, $(z,r)\in\mathcal{E}$ and $f(z,r)>0$, do: $\mathcal{C}=\mathcal{C}/\{r\}$,
$\mathcal{C}_{c}=\mathcal{C}_{c}/\{r\}$.
\item \textbf{Step 3}: For each $z\in\mathcal{U}_{c}$ such that $\exists$
$r\in\mathcal{C}$, $(z,r)\in\mathcal{E}$ , do: $\mathcal{U}=\mathcal{U}\cup\{z\}$,
$\mathcal{U}_{c}=\mathcal{U}_{c}/\{z\}$.
\item \textbf{Step 4}: Iterate between \textbf{Step 2} and \textbf{Step
3} until no vertices can be added to $\mathcal{C}$ or $\mathcal{U}$.
$\overline{\Omega}_{sub}$ is given by $\overline{\Omega}_{sub}=\{(j,i,p,q):c_{jipq}\in\mathcal{C}\}.$
\hfill \IEEEQED
\end{itemize}

We have $\mathcal{U}=\{u_{jp},v_{iq},\forall(j,i,p,q)\in\overline{\Omega}_{sub}\}$.
Furthermore, as the max-flow graph is symmetric w.r.t. $q$ and $d\mid c(s,u_{jp})$,
$\forall j,p$, we must have $f(s,z)=c(s,z)$, $\forall z\in\mathcal{U}$
(otherwise there exist further augmenting paths \cite{cormen2001introduction}
in the graph). Hence, 
\begin{eqnarray*}
\sum_{(j,p):\;(j,i,p,q)\in\overline{\Omega}_{sub}}U_{jp}+\sum_{(i,q):\;(j,i,p,q)\in\overline{\Omega}_{sub}}V_{iq} & = & \sum_{(j,i,p,q)\in\overline{\Omega}_{sub}}\left(f(c_{jipq},t)\right)\\
 & \leq & \sum_{\underset{\{x,y,m,n=\{1,\cdots d\}\}}{(j,i,p,q)\in\overline{\Omega}_{sub}/}}c(c_{jipq},t)+\sum_{n=1}^{d}f((c_{xymn},t))\\
 & < & \sum_{(j,i,p,q)\in\overline{\Omega}_{sub}}1,
\end{eqnarray*}
which contradicts condition (\ref{eq:sub_system}).\hfill \IEEEQED

Via the above converse-negative proposition, Lemma \ref{Equivalent-feasibility-problem}
is proved.

\subsection{\label{sub:Complexity-of-Feasibility}Complexity of Feasibility Checking}

If $\overline{M}<\sum_{j,i\in\Omega_{j}^{IV}}d_{j}^{0}d_{i}$, where
$\overline{M}$ is in Theorem \ref{Sufficient-Feasibility-Condition},
then the IA problem is infeasible under the current feedback profile
as condition 3) in Theorem 1 is violated. If $\overline{M}\geq\sum_{j,i\in\Omega_{j}^{IV}}d_{j}^{0}d_{i}$,
we check the linear independence of the row vectors $\{\mathbf{X}_{ji}\}$
by checking whether the determinant of matrix $\overline{\mathbf{X}}$
is nonzero with complexity $\mathcal{O}(\overline{M}^{3})=\mathcal{O}\left((KN)^{3}\right)$
\cite{bunch1974triangular}:
\[
\overline{\mathbf{X}}=\left[\begin{array}{cc}
\left[\begin{array}{ccc}
\cdots & \mathbf{X}_{ji}^{T} & \cdots\end{array}\right]_{j,i\in\Omega_{j}^{IV}} & (\mathbf{X}^{[c]})^{T}\end{array}\right]\in\mathbb{C}^{\overline{M}\times\overline{M}}
\]
where $\mathbf{X}^{[c]}$ is a $\left(\overline{M}-\sum_{j,i\in\Omega_{j}^{IV}}d_{j}^{0}d_{i}\right)\times\overline{M}$
random matrix independent of $\{\mathbf{X}_{ji}\}$. Note the row
vectors of $\{\mathbf{X}_{ji}:\forall j,i\in\Omega_{j}^{IV}\}$ are
independent if and only if $\textrm{det}(\overline{\mathbf{X}})\neq0$.

\subsection{\label{sub:Proof-of-Theorem-performance}Proof of Theorem \ref{Performance-on-special}}

We sketch the proof due to page limit. In the initial step in Algorithm
2, $\mathcal{L}$ becomes: $N_{i}^{s}=Kd$, $M_{i}^{s}=\min(M,Kd)$,
$\forall i$, $\Omega_{j}^{I}=\Omega_{j}^{II}=\Omega_{j}^{III}=\emptyset$,
$\Omega_{j}^{IV}=\{1,\cdots j-1,j+1,\cdots K\}$, $\forall j$. After
that, (a) if $d\leq\frac{M}{K}$, Algorithm 2 will update $\mathcal{L}$
by adopting strategy $S^{I}(j,i)$ for all $j,i,j\neq i$ and then
adopting $S^{V}(j)$, for all $j$ until we obtain the desired result;
(b) if $\frac{M}{K}<d\leq M$, $d\mid M$, Algorithm 2 will update
$\mathcal{L}$ through the following four stages sequentially: 
\begin{itemize}
\item A: keep adopting strategy $S^{IV}(i)$ for all $K^{s}+1\leq i\leq K$,
and we obtain the updated $\mathcal{L}$: $N_{i}^{s}=Kd$ for $1\leq i\leq K^{s}$
and $N_{i}^{s}=d$ for $K^{s}+1\leq i\leq K$, $M_{i}^{s}=M$, $\forall i$,
$\Omega_{j}^{II}=\{t:t=K^{s}+1,\cdots K,t\neq j\}$, $\Omega_{j}^{I}=\Omega_{j}^{III}=\emptyset$,
$\Omega_{j}^{IV}=\{1,\cdots j-1,j+1,\cdots K\}/\Omega_{j}^{II}$,
$\forall j$. 
\item B: keep adopting strategy $S^{III}(j,i)$ for all $1\leq j,i\leq K^{s}+1,i\neq j$,
and we obtain the updated $\mathcal{L}$: $N_{i}^{s}$, $M_{i}^{s}$,
$\Omega_{j}^{I}$, $\Omega_{j}^{II}$ the same as in A, $\Omega_{j}^{III}=\{t:t=1,\cdots,K^{s},t\neq j\}$,
$\Omega_{j}^{IV}=\{1,\cdots j-1,j+1,\cdots K\}/(\Omega_{j}^{II}\bigcup\Omega_{j}^{III})$,
$\forall j$.
\item C: keep adopting strategy $S^{V}(j)$ for all $K^{s}+2\leq j\leq K$
with each $S^{V}(j)$ repeating $d$ times, and we obtain the updated
$\mathcal{L}$: $N_{i}^{s}=Kd$, $M_{i}^{s}=M$ for $1\leq i\leq K^{s}$
and $N_{i}^{s}=d$, $M_{i}^{s}=M-d$ for $K^{s}+1\leq i\leq K$, $\forall i$,
$\Omega_{j}^{I}$, $\Omega_{j}^{II}$, $\Omega_{j}^{III}$, $\Omega_{j}^{IV}$
the same as in B. 
\item D: keep adopting strategy $S^{III}(j,i)$ for all $K^{s}+2\leq j\leq K,1\leq i\leq K^{s}+1$,
and we obtain the final feedback profile $\mathcal{L}$. 
\end{itemize}
Substitute the final $\mathcal{L}$ into (\ref{eq:sum_feedback_dimension_expression}),
we obtain the associated feedback dimension $D_{p}$.

\bibliographystyle{IEEEtran}
\bibliography{Zf_IA_ref}

\end{document}